\documentclass[11pt]{article}
\usepackage{fullpage}
\usepackage[top=1in,left=1in,right=1in,bottom=1in]{geometry}
\usepackage{parskip}
\pdfoutput=1
\usepackage{graphicx} 
\usepackage{sidecap}
\usepackage{hyperref}
\hypersetup{
        colorlinks=true,  
        citecolor=blue,   
        linkcolor=blue,    
        urlcolor=cyan     
    }

\usepackage{blkarray}
\usepackage{tikz}
\usetikzlibrary{matrix}
\usetikzlibrary{positioning}
\usepackage{amsmath, amssymb, amsthm,bbm, mathtools}
\usepackage{thmtools, thm-restate}
\newcommand{\email}[1]{\href{mailto:#1}{\texttt{#1}}}

\usepackage{authblk}

\usepackage{algorithmic}
\usepackage[ruled]{algorithm}
\usepackage{graphicx}
\usepackage{color}
\usepackage[dvipsnames]{xcolor}
\usepackage{enumerate}
\usepackage{thm-restate}

\usepackage{tcolorbox}

\usepackage[nameinlink]{cleveref}
\usepackage{natbib}

\usepackage{nicefrac}

\newcommand{\R}{\mathbb{R}}

\newcommand{\calD}{\mathcal{D}}

\newcommand{\calN}{\mathcal{N}}

\newtheorem{theorem}{Theorem}[section]
\newtheorem*{theorem*}{Theorem}

\newtheorem{proposition}[theorem]{Proposition}
\newtheorem*{proposition*}{Proposition}
\newtheorem{lemma}[theorem]{Lemma}
\newtheorem*{lemma*}{Lemma}
\newtheorem{corollary}[theorem]{Corollary}
\newtheorem*{conjecture*}{Conjecture}
\newtheorem{fact}[theorem]{Fact}
\newtheorem*{fact*}{Fact}

\newtheorem*{hypothesis*}{Hypothesis}
\newtheorem{conjecture}[theorem]{Conjecture}
\newtheorem{itheorem}[theorem]{Informal Theorem}

\newtheorem{claim}{Claim}[section] 
\newtheorem*{claim*}{Claim}

\theoremstyle{definition}
\newtheorem{definition}[theorem]{Definition}

\newtheorem{question}[theorem]{Question}

\newtheorem{problem}[theorem]{Problem}

\newtheorem*{question*}{Question}

\theoremstyle{remark}

\newtheorem*{remark*}{Remark}

\newcommand{\eat}[1]{}

\newcommand{\vol}{\mathrm{vol}}

\newcommand{\norm}[1]{\lVert #1 \rVert}

\newcommand{\iprod}[1]{\langle#1\rangle}

\newcommand{\Esymb}{\mathbb{E}}
\newcommand{\Psymb}{\mathbb{P}}

\DeclareMathOperator*{\E}{\Esymb}
 \DeclareMathOperator*{\Var}{\mathrm{Var}}
 \DeclareMathOperator*{\ProbOp}{\Psymb}
 
 \DeclareMathOperator*{\argmin}{argmin}

\renewcommand{\Pr}{\ProbOp}

\newcommand{\diag}{\text{diag}}

 \newcommand{\eps}{\varepsilon}
\renewcommand{\epsilon}{\varepsilon}

\newcommand{\poly}{\mathrm{poly}}

\newcommand{\epst}{\eps_{\star}}
\newcommand{\gammast}{\gamma_{\star}}

\DeclareMathOperator{\dist}{dist}

\newcommand{\wta}{\widetilde{a}}

\newcommand{\wtw}{\widetilde{w}}

\newif\ifnotes\notesfalse

\ifnotes
\usepackage{color}
\definecolor{mygrey}{gray}{0.50}
\newcommand{\notename}[2]{{\textcolor{blue}{\footnotesize{\bf (#1:} {#2}{\bf ) }}}}

\newcommand{\vnote}[1]{{\notename{Vaidehi}{#1}}}

\newcommand{\anote}[1]{{\notename{Aravindan}{#1}}}
\else

\newcommand{\notename}[2]{{}}

\newcommand{\enote}[1]{}
\newcommand{\vnote}[1]{}
\newcommand{\bnote}[1]{}
\newcommand{\anote}[1]{}

\fi

\usepackage{soul}

\title{Learning Confidence Ellipsoids and Applications to Robust Subspace Recovery}
\author{}
\author{    
    Chao Gao\thanks{Department of Statistics, University of Chicago, Chicago, USA, \email{chaogao@uchicago.edu}} , 
    Liren Shan\thanks{Toyota Technological Institute at Chicago, Chicago, USA, \email{lirenshan@ttic.edu}} ,  Vaidehi Srinivas\thanks{Department of Computer Science, Northwestern University, Evanston, USA, \email{vaidehi@u.northwestern.edu}} ,  Aravindan Vijayaraghavan\thanks{Department of Computer Science, Northwestern University, Evanston, USA, \email{aravindv@northwestern.edu}}
    }

\date{}

\begin{document}
\maketitle
\begin{abstract}
We study the problem of finding confidence ellipsoids for an arbitrary distribution in high dimensions. Given samples from a distribution $\mathcal{D}$ and a confidence parameter $\alpha$, the goal is to find the smallest volume ellipsoid $E$ which has probability mass $\Pr_{\mathcal{D}}[E] \ge 1-\alpha$. Ellipsoids are a highly expressive class of confidence sets as they can capture correlations in the distribution, and can approximate any convex set. This problem has been studied in many different communities.  In statistics, this is the classic {\em minimum volume estimator} introduced by Rousseeuw as a robust non-parametric estimator of location and scatter. However in high dimensions, it becomes NP-hard to obtain any non-trivial approximation factor in volume when the condition number $\beta$ of the ellipsoid (ratio of the largest to the smallest axis length) goes to $\infty$. This motivates the focus of our paper: 
%
\noindent {\em can we efficiently find confidence ellipsoids with volume approximation guarantees when compared to ellipsoids of bounded condition number $\beta$?}

Our main result is a polynomial time algorithm that finds an ellipsoid $E$ whose volume is within a $O(\beta)^{\gamma d}$ multiplicative factor 
of the volume of best $\beta$-conditioned ellipsoid while covering at least $1-O(\alpha/\gamma)$ probability mass for any $\gamma \in (0,1)$.  In particular, setting $\gamma = o(1)$, this gives a $O(\beta)^{o(d)}$ volume approximation, with a multiplicative loss in miscoverage. 
We complement this with a computational hardness result that shows that such a dependence on $\beta$ in the volume approximation factor seems necessary, even with some slack in coverage.   
The algorithm and analysis uses the rich primal-dual structure of the minimum volume enclosing ellipsoid and the geometric Brascamp-Lieb inequality. This natural iterative algorithm also gives a new subroutine for robust estimation problems.  
As a consequence, we obtain the first polynomial time algorithm with approximation guarantees on worst-case instances of the robust subspace recovery problem.   
\end{abstract}

\thispagestyle{empty} 

\newpage 
\thispagestyle{empty} 
\tableofcontents

\newpage
\setcounter{page}{1}
\section{Introduction}

Finding small confidence sets is a basic and ubiquitous task in statistics and data analysis.  
Given independent samples drawn from the distribution $\mathcal{D}$ over $\R^d$ and a confidence parameter $\alpha \in (0,1)$, the goal is to find the set $S$ in a natural class $\mathcal{C}$ of confidence sets that minimizes volume while achieving the desired coverage of $1-\alpha$:
\begin{equation} \min \vol(S) \quad \text{s.t. } S \in \mathcal{C}, \quad \mathbb{P}_{\calD}[S] \ge 1-\alpha. \label{eq:popln}
\end{equation}
Here $\vol(S)$ denotes the volume or Lebesgue measure of the set $S$. This is a central and well-studied problem in statistics, with applications to estimating density level sets, support estimation, uncertainty quantification, robust estimation and conformal prediction~\cite{einmahl1992generalized, polonik1997minimum, rousseeuw1985multivariate, polonik1999concentration, garcia2003level, scott2005learning, gao2025volume}.  
When the VC dimension of the class $\mathcal{C}$ is bounded, this task becomes statistically tractable by solving the empirical version of the problem on a finite number of samples~\cite{polonik1999concentration, scott2005learning}. This problem has been studied for several natural classes including balls, ellipsoids, rectangles or products of intervals, and even some more general non-convex sets, but it becomes algorithmically challenging in high dimensions. 

Among natural classes of confidence sets, ellipsoids occupy a privileged position. 
Confidence ellipsoids are a highly expressive family of confidence sets, as they naturally adapt to correlations among high-dimensional outputs and highlight directions of high or low uncertainty. They have been successfully applied in practical settings with structured outputs --- for instance, in time-series forecasting, where prediction errors across timesteps are highly correlated 
~\cite{conformalellipsoid2022, conformalellipsoid2024xu}.  
Moreover, confidence ellipsoids are universal approximators of {\em all convex confidence sets} due to the famous John's theorem~\cite{john1948extremum}.  

Recall that an ellipsoid $E$ is parameterized by a center $c\in \R^d$ and a positive definite matrix $M$ and given by $E=\{x: \norm{M^{-1/2}(x-c)}_2 \le 1\}$. 
The family of ellipsoids in $\R^d$ has VC dimension $\binom{d+1}{2}+d$. Given any arbitrary distribution $\calD$, we can draw $n = O(d^2)$ i.i.d.\ samples from $\calD$, and solve the following empirical version of the problem.  For $n$ \emph{arbitrary} points $a_1, a_2, \dots, a_n \in \R^d$, find
\begin{align}\label{eq:intro:empirical}
\argmin_{E \in \mathcal{E}} \vol(E) \text{ s.t. } \Big|E \cap \{a_1, \dots, a_n \} \Big| \ge (1-\alpha)n, 
\end{align}

\noindent where $\mathcal{E}$ is the family of ellipsoids in $d$ dimensions. Any ellipsoid $E$ that contains at least $(1-\alpha)n$ of the points also achieves a coverage of $1-\alpha - O(\sqrt{\frac{d^2}{n}})$ on $\calD$, via uniform convergence for set families of bounded VC dimension.  While there is a distributional flavor to the population version of the problem in \eqref{eq:popln} (with $\mathcal{C}=\mathcal{E}$), the distributional and empirical problems are actually equivalent.  Given any arbitrary set of $n$ points $a_1, \dots, a_n \in \R^d$, one can just consider $\calD$ to be the uniform probability distribution on these $n$ points. This emphasizes the worst-case nature of the problem, as the distribution $\calD$ is arbitrary. This is also crucial in many applications where the method is used to estimate uncertainty of an (unknown) arbitrary distribution.  

Problem \eqref{eq:intro:empirical} has been an important object of study in several different communities. In statistics, this is the famous {\em minimum volume estimator} (MVE) of Rousseeuw~\cite{rousseeuw1985multivariate}, that is studied as a robust non-parametric estimator for location and scatter parameters of a distribution. It is preferred for several desirable properties including affine equivariance (they behave well under affine transformations of the data), and having the {\em maximum breakdown value}\footnote{The breakdown value is the fraction of outliers the estimator can tolerate while providing a meaningful estimate.} among all robust estimators of scatter and location (see the excellent survey by \cite{van2009minimum} for details). 
Solving this problem exactly is known to be NP-hard~\cite{thorsten, AhmadiML2014}. This is also been studied as the {\em robust minimum volume ellipsoid}, {\em minimum volume ellipsoid with outliers}, and {\em minimum $k$-enclosing ellipsoid problem} in optimization,  applied mathematics and computational statistics, where many heuristics have been developed.
These methods have been used as subroutines for applications in control theory, dynamical systems, and uncertainty quantification~\cite{Agullo1996, AhmadiML2014, Ahipasaoglu2014, CalafioreG04, lindemann2024,  henderson2025adaptiveinferencerandomellipsoids}.  
%

However, despite several decades of work, to the best of our knowledge, there is no known polynomial time algorithm for learning small confidence ellipsoids. 

\vspace{5pt}
\noindent {\bf Question: }{\em Can we design polynomial time algorithms that learn confidence ellipsoids with rigorous volume guarantees?}

\vspace{5pt}

 When $\alpha=0$, this is the classic minimum volume enclosing ellipsoid problem (MVEE) which can be solved in polynomial time using semi-definite programming~\cite{boyd2004convex, Toddbook}. This is also the famous L\"owner-John ellipsoid for the convex hull of the given points~\cite{john1948extremum}. For general $\alpha \in (0,1)$, very recent work of \cite{gao2025confidence} studied the special case of \eqref{eq:intro:empirical} for Euclidean balls, which is already known to be NP-hard. They give polynomial time algorithms that achieve volume guarantees that are competitive with that of the optimal solution up to a multiplicative factor, while relaxing the confidence parameter $\alpha$ by a small amount. 
 
  The algorithmic problem~\eqref{eq:intro:empirical} 
  is significantly more challenging 
  compared to balls. 
  A special case is {\em robust subspace recovery} which is the computationally hard problem of finding a subspace of dimension at most a prescribed $(1-\gamma)d$ that contains an $1-\alpha$ fraction of points~\cite{HardtMoitra13}. The NP-hardness of the robust subspace recovery problem directly implies that  
it is NP-hard to get {\em any non-trivial volume approximation} for \eqref{eq:intro:empirical} i.e., distinguish when the minimum volume ellipsoid containing a $1-\alpha$ fraction of points has $\vol=0$ versus $\vol>0$ (see Appendix~\ref{app:nphardness}). 
However, this hard instance is pathological, where the optimal ellipsoid is ill-conditioned, with one of the axes lengths being $0$ (or close to $0$). This motivates restricting our attention to ellipsoids with bounded {\em condition number}. 

Given an ellipsoid $E = \{x: \norm{M^{-1/2}(x-c)}_2 \le 1\}$, the condition number $\beta(E) \in [1,\infty)$ is the ratio of maximum to minimum axes lengths:
\begin{equation}\label{eq:conditionnumber}
\beta(E) \coloneq \frac{\max_{i \in [d]}\lambda_i(M^{1/2})}{\min_{i \in [d]} \lambda_i(M^{1/2})} = \sqrt{\frac{\max_{i \in [d]}\lambda_i(M)}{\min_{i \in [d]} \lambda_i(M)}},
\end{equation}
\noindent where $\lambda_1(M), \dots, \lambda_d(M)$ are the eigenvalues of $M$. The condition number $\beta(E)=1$ for a Euclidean ball, and $\beta \to \infty$ as the ellipsoid becomes supported on a proper subspace of $\R^d$. 
It will be instructive to think of $\beta$ as being a large constant or polynomial in the dimension $d$. The main question that we study in this paper is the following.

\begin{tcolorbox}
\begin{question*}[Minimum $\beta$-Conditioned Confidence Ellipsoids]\label{qn:ellipsoid}
Suppose there exists a $\beta$-conditioned ellipsoid $E^\star$ that contains $(1-\alpha)$ fraction of the input points $a_1, \dots, a_n \in \R^d$.  Can we, in polynomial time, find an ellipsoid $E$ containing an $1-\alpha'$ fraction of the points whose volume is $\Gamma$-competitive with $E^\star$ i.e.,
\begin{equation}
\vol(E)^{1/d} \le \Gamma \cdot \vol(E^\star)^{1/d},
\end{equation}
\end{question*}
where the missed coverage $\alpha'=f(\alpha)$ goes to $0$ as $\alpha \to 0$?
\end{tcolorbox}

Here the approximation factor $\Gamma(\beta,d, \alpha)$ is a function of $\beta, d, \alpha$ and is measured with respect to $\vol^{1/d}$ following the convention in \cite{gao2025confidence}, and the corresponding approximation factor for volume becomes $\Gamma^d$. 



\subsection{Our Results} \label{sec:results}


Our main result is a polynomial time algorithm that achieves a rigorous volume guarantee for ellipsoids whose condition number $\beta$ is bounded. 

\begin{theorem}[Algorithm for $\beta$-conditioned ellipsoids]\label{thm:main}
For any parameters $\alpha \in (0,1)$, $\gamma \in (0,1)$, $\beta \ge 1$, 
there exists a universal constant $c>0$ such that the following holds. There is an algorithm that given a set of $n$ points $A=\{a_1, \dots, a_n\} \in \R^d$ for which there exists an ellipsoid $E^\star \subset \R^d$ that is $\beta$-conditioned and $|E^\star \cap A| \ge (1-\alpha)n$ , finds in polynomial time an ellipsoid $\widehat{E} \subset \R^d$  with $\Gamma=\beta^\gamma$ i.e.,
\begin{equation}\label{eq:main:volume}
\vol^{1/d}(\widehat{E}) \le (4\beta)^{2\gamma \left(1 + \frac{1}{d} \right)} \cdot \vol^{1/d}(E^\star), \text{ and encloses at least } |\widehat{E} \cap A| \ge \Big(1-\frac{c\alpha}{\gamma}\Big)n  \text{ points}.
\end{equation}
\end{theorem}
\vnote{do we want to use thm-restate with the formal version here?}\anote{Sure, please feel free to take executive decisions.}
(See \Cref{thm:approx-optimal-confidence-ellipsoid} in \Cref{sec:algorithm} for more details.)
This immediately implies a corresponding guarantee for the distributional version, using the VC dimension bound on ellipsoids. 

\begin{corollary}[Distributional version of \Cref{thm:main}]\label{cor:main}
For any parameters $\alpha \in (0,1)$, and $\gamma \in (0,1), \beta \ge 1$, there is an algorithm that given 
$n$ i.i.d. samples from an arbitrary distribution $\calD$ over $\R^n$ with an (unknown) $\beta$-conditioned ellipsoid $E^\star \subset \R^d$ satisfying $\Pr_\calD[E^\star] \ge (1-\alpha)$ , finds in polynomial time an ellipsoid $\widehat{E} \subset \R^d$  with
\begin{equation}\label{eq:main:volume_distribution}
\vol^{1/d}(\widehat{E}) \le (4\beta)^{2\gamma \left(1 + \frac{1}{d} \right)} \cdot \vol^{1/d}(E^\star), \text{ and } \Pr_\calD[\widehat{E}] \ge 1-O\Big(\frac{\alpha}{\gamma} + \sqrt{\frac{d^2}{n}} \Big).
\end{equation}
\end{corollary}

\paragraph{Implication: $\Gamma = \beta^{o(1)}$ approximation.} Consider the setting when $\alpha = o(1)$ (or a very small constant). We can compare against an algorithm that uses as a black-box the guarantee given by \cite{gao2025confidence} for Euclidean balls. Since the ellipsoid $E^\star$ is $\beta$-conditioned, there is also a Euclidean ball whose volume is within a $\beta^d$ multiplicative factor of it. The result of \cite{gao2025confidence} computes an ellipsoid whose volume is within a $2^{o(d)}$ factor to the volume of the best ball (while also incurring a small loss in coverage). This leads to a volume approximation factor of $2^{o(d)}\beta^d$. In other words, this achieves an approximation factor $\Gamma \approx \beta$. 

In comparison, our algorithm achieves a guarantee of $\Gamma=\beta^\gamma$ i.e., a volume approximation of $\beta^{\gamma d}$, with miscoverage $\alpha/\gamma$. This achieves a $\Gamma=\beta^{o(1)}$ factor for any choice of $\gamma = o(1)$. For example, by setting $\gamma=\sqrt{\alpha}$, the ellipsoid encloses at least $(1-O(\sqrt{\alpha}))$ fraction of the points, and achieves 
Alternatively, by setting $\gamma = \omega( 1/\log \beta)$, we can achieve an approximation factor $\Gamma = \exp(o(d))$, which is independent of $\beta$.

\vnote{03/31 added w.r.t. reviewer feedback:} As we will see later, the $\Gamma = \beta^{o(1)}$ guarantee is key to getting an algorithm for robust subspace recovery.  This gives an example of how this guarantee is qualitatively different and more challenging than the weaker $\Gamma = \beta$ guarantee.  

Moreover, our algorithm deviates significantly from the approach of \cite{gao2025confidence} which leveraged a new connection 
to robust high-dimensional estimation. In contrast, we design and analyze a natural iterative algorithm based on the rich primal-dual structure of the {\em minimum volume enclosing ellipsoid} problem, and the {\em D-optimal design} problem,\footnote{The dual to the minimum volume enclosing ellipsoid SDP is the well-studied {\em D-optimal design} problem or equivalently the {\em G-optimal design} problem in statistics and experimental design, that tries to weight the covariates of a linear regression problem to minimize the variances of the predictions~\cite{Toddbook}. See also \cite{nikolov2015max, singh2016maxdet} for the ``integral'' version of the problem and other interesting variants. } and serves as a potential subroutine for other robust estimation tasks. 



 Our next result shows that the dependence on $\beta$ in the factor $\Gamma$ (and the volume approximation) is unavoidable (up to a constant factor in the exponent), even when we allow for some slack in the missed coverage $\alpha$. This is assuming the well-studied {\em small set expansion (SSE)} conjecture of Raghavendra and Steurer~\cite{raghavendra2010graph}, which is closely related to the unique games conjecture~\cite{khot2003UGC}. 
 In fact the hardness holds for $\Gamma = \beta^{\Omega(1)}$; hence achieving a factor $\Gamma = \beta^{o(1)}$ requires some loss in coverage. 

\anote{Liren, please verify this statement. }
\begin{itheorem}[SSE-hardness]\label{ithm:SSE-hard}
 Assuming the Small-Set Expansion conjecture, there exist small constants $\alpha, \delta>0$ and $\beta_0 \ge 1$ such that for {any} (sufficiently) large $\beta \ge \beta_0$, no
 polynomial-time algorithm when given as input $n$ points with a $\beta$-conditioned ellipsoid containing a $1-\alpha$ fraction of them can achieve a factor $\Gamma < \beta^{\delta-o(1)}$ with coverage $1-(1+\delta)\alpha$.
\end{itheorem}

See \Cref{thm:SSE-hard} in \Cref{sec:SSE-hard} for a formal statement. This hardness result is based on the hardness for robust subspace recovery due to Hardt and Moitra~\cite{HardtMoitra13}, and already suggests a strong connection between the two problems. 




\subsubsection{Robust Subspace Recovery and Other Applications}

Our algorithm from \Cref{thm:main} gives a new algorithmic tool for robust estimation problems in high-dimensions. We show that we can use our algorithm to obtain polynomial time algorithmic guarantees for the challenging {\em robust subspace recovery} problem.  
Recall that we are given as input $n$ points $a_1, \dots, a_n \in \R^d$ (which we can assume are unit vectors), an outlier fraction $\alpha \in (0,1)$, and a closeness parameter $\eps \in (0,1)$.  The goal is to find a subspace $S \subset \R^d$ (passing through the origin) of smallest possible dimension (or largest co-dimension) that is $\eps$-close in $\ell_2$ distance to at least $(1-\alpha)n$ points. 
This is a well-studied problem in statistics, with motivation as a variant of PCA that finds a low-dimensional subspace even in the presence of outliers~\cite{subspacerecoverySurvey, HardtMoitra13}. 
Due to the NP-hardness of the problem~\cite{khachiyan}, previous works that provide polynomial time guarantees need to make  structural assumptions on either the outliers, or the inliers (points on the subspace). 

\paragraph{Prior algorithmic results on robust subspace recovery.} The work of Hardt and Moitra \cite{HardtMoitra13} assumes genericity (similar to a smoothed analysis assumption) on the outliers and gives elegant polynomial time algorithms when the dimension of the subspace $k$ containing $(1-\alpha)n$ points satisfies $k \le d(1-\alpha)$. They complement their result with a hardness result in the {\em worst-case setting} when $k \ge (1-\alpha) d$ assuming the SSE conjecture, even in the noiseless setting where $\eps=0$ \cite{HardtMoitra13}. 
The work of \cite{BCPV} improved the condition on $\alpha$ in the smoothed analysis setting, and also provided some noise tolerance $\eps \approx 1/\poly(n)$. More recent works based on sum-of-squares (SoS) relaxations assume distributional assumptions (e.g., Gaussianity or rotational invariance) on the points in the subspace~\cite{bakshi2021list} with stronger noise tolerance.
In contrast we study the worst-case version of the problem, and make no assumptions on either the inliers or outliers. We are unaware of any non-trivial algorithmic guarantee in the worst-case setting of the problem prior to our work. We show that the algorithm for learning small volume confidence ellipsoids can be used as a subroutine to provide algorithmic guarantees for robust subspace recovery in the worst-case. 

\anote{12/18:Removed condition on $n$.}
\begin{itheorem}\label{ithm:subspacerecovery}
For any $\gamma \in (0,1)$, there is a polynomial time algorithm, that given as input points $A=\{a_1, \dots, a_n\} \in \R^d$ with a subspace $S^\star$ of dimension $(1-8\gamma) d$ that is $\epst \le \eps^{4/\gamma}$ close to at least $(1-\alpha)n$ of the given points, finds w.h.p.\ a subspace $S \subset \R^d$ of dimension at most $(1-\gamma)d$ such that at least $(1-O(\alpha/\gamma))n$ points are $\eps$ close to $S$. 
\end{itheorem}

See \Cref{thm:subspacerecovery} for the formal statement. 
When we know that there are many points that lie exactly on the subspace $S$, we have $\epst=0$. We can pick $\eps$ to be whatever desired accuracy (that can be arbitrarily close to $0$) 
and apply the above algorithm to find a subspace that is $\eps$-close to the desired number of points. On the other hand, we also get guarantees when the desired accuracy $\eps$ is a small constant (for e.g., $\gamma$ is a constant and $\epst$ is a sufficiently small constant) or $\eps$ is inverse polynomial in $d^\gamma$ (when $\epst \le 1/\poly(d)$). 

\vnote{03/31 added w.r.t to reviewer feedback:}
To use this approach, it is crucial that we use an algorithm that can approximate a confidence ellipsoid up to a factor better than \(\Gamma  = \beta\).  Our algorithm uses the fact that the points that lie close to the true underlying subspace must lie in a thin slab.  Thus, they also lie in an ellipsoid that is thin in many directions, corresponding to the co-dimension of the subspace.  Our algorithm applies a small random perturbation to all of the points, ensuring that we can control the width of this ellipsoid in all of the thin directions.  Then, we argue that a guarantee of \(\Gamma = \beta^\gamma\) (i.e. \(\beta^{\gamma d}\) approximation in volume) to the ellipsoid can only blow up \(\approx \gamma d\) of these directions.  To recover subspaces of nontrivial dimension $< d$, we use that $\gamma d \ll d$ so we do not blow up too many directions.  
This gives an example why the $\Gamma = \beta^{o(1)}$ guarantee is qualitatively different and more challenging than the trivial $\Gamma = \beta$ guarantee.  
To the best of our knowledge, this provides the first rigorous polynomial time guarantee for robust subspace recovery in the worst-case setting.



\paragraph{Statistical Applications.}  Our algorithm for learning confidence ellipsoids has some direct applications to other related problems.  One implication of our algorithm is that it gives an algorithm for learning \emph{convex confidence sets} with bounded condition number that are of approximately minimum volume (\Cref{thm:convex-via-john}).  This is because any convex set can be well-approximated by an ellipsoid, and our algorithm can provide a good approximation to the ellipsoid.  An interesting aspect of this guarantee is that the family of convex sets has infinite VC dimension.  We know that bounded VC dimension is a sufficient condition for a family of confidence sets to be statistically efficiently learnable.  However, our result suggests that bounded VC dimension is not a necessary condition for learning an approximately volume optimal confidence set in the family. 

Another implication of our result is for high-dimensional conformal prediction (\Cref{thm:conformal-ellipsoid}).  In conformal prediction, the goal is to use training examples to construct a prediction set that contains an unseen test example with probability at least \(1 - \alpha\).  An algorithm for finding small confidence sets on arbitrary samples directly implies an efficient conformal prediction algorithm. Please see \Cref{sec:applications} for details.


\subsection{Technical Overview} \label{sec:overview}

\subsubsection{Approximate Minimum Volume Confidence Ellipsoid}
Given a set of points \(A = \{a_1, \dots, a_n\} \subseteq \mathbb{R}^d\), 
our goal is to find a small ellipsoid that contains \((1- \alpha)\)-fraction of the points in \(A\).  Let \(E^\star\) be the optimal ellipsoid that encloses \((1 - \alpha)n\) points of \(A\), and has condition number at most \(\beta\).  We will refer to the points enclosed by \(E^\star\) as the \emph{inliers}, and the \(\alpha n\) points of \(A\) that are not enclosed by \(E^\star\) as the \emph{outliers}. 

\paragraph{Origin-centered setting and SDP formulation.}

For the purpose of the technical overview, we will focus on the {\em origin-centered setting}, when the center of the ellipsoid is fixed to be the origin. This special case is already challenging and well-motivated, as it captures the application to the robust subspace recovery problem. Moreover it will inspire the main ideas in the algorithm and analysis --- the general setting where the center is not the origin is handled by solving an appropriate origin-centered problem in $d+1$ dimensions. 
In \Cref{sec:algorithm}, we present the full algorithm for  the general setting where the center of the ellipsoid can be arbitrary. In the rest of this overview, we will present the main algorithmic ideas, and give a proof sketch of the entire analysis. 

For the classic problem of finding the {\em minimum volume enclosing ellipsoid} (origin-centered) containing \emph{all} of the points in \(A\), then there is a very well-studied formulation of this problem as a semidefinite program (SDP): 
\anote{Changed $M$ to $Q$}
\begin{align*}
    \max_Q ~& \log \det Q \\
    \text{s.t. }& \forall i, ~a_i^\top Q a_i \le 1 & \text{(MVEE-OC)}\\
    & Q \succeq 0.
\end{align*}
Note that the volume of the resulting ellipsoid \(E = \{x: x^\top Q x \le 1\}\) is \(\kappa_d \det(Q^{-1/2})\), where $\kappa_d>0$ is the volume of the unit radius ball in $d$-dimensions. For the purposes of this paper, we can ignore this fixed constant $\kappa_d$, as all the $d$-dimensional volumes will be scaled by the same amount. This SDP can be solved in polynomial time up to desired accuracy $\eta$ in polynomial time in $d, n, \log(1/\eta)$. 
The goal is now to choose \(\alpha n\) points of \(A\) to leave out, such that it brings down the objective value of the above SDP (MVEE-OC) as much as possible. The above SDP seems challenging to work with directly when we only need to enclose $1-\alpha$ of the distribution. 

Which points contribute to increasing the volume of the enclosing ellipsoid the most?  Intuitively, this is a problem that has to be addressed somehow by the SDP itself, as these points are the ones that the solution has to work hardest to cover.  
This motivates us to look at the dual formulation, which has a variable that captures how hard it is to satisfy the constraint from each point.  

\paragraph{Dual SDP and D-Optimal Design.} The dual of the minimum volume enclosing ellipsoid problem is the classic {\em D-optimal design} problem,   
which is (a scaled version of) the problem of finding a distribution supported on the points \(a_i\), that maximizes the determinant of the covariance or second moment matrix (centered at the origin). 
\begin{align*}
    \max_{w_i, \forall i \in [n]} ~& \log \det \left(\sum_{i = 1}^n w_i a_i a_i^\top \right) \\
    \text{s.t. }& \forall i, ~w_i \ge 0  & \text{(dual-MVEE)}\\
    & \sum_{i = 1}^n w_i \le d.
\end{align*}
The weights \(w_i\) corresponding to the optimal solution of the dual can be shown to satisfy \(w_i \le 1\).

Intuitively, far-away outliers will contribute the most to increasing the volume of the enclosing ellipsoid. The dual SDP will likely place a significant fraction of weight on outliers that lie far outside the desired small ellipsoid.  Another interpretation of this is via complementary slackness.  The points that are assigned non-zero weight in the dual solution are exactly the points that lie on boundary of the optimal ellipsoid in the primal SDP solution.  Deleting a point with $w_i=0$ i.e., one that is contained in the interior of the ellipsoid cannot, by itself, bring down the volume of the enclosing ellipsoid.  
Hence the weights \((w_i: i \in [n])\) can be interpreted as a measure of how important the point is to the volume of the enclosing ellipsoid.  

\paragraph{Iterative algorithm based on dual and progress in outlier removal.} The above intuition suggests the following natural iterative algorithm.  
 Our algorithm proceeds by iteratively solving the dual program (dual-MVEE) on the set of points $A=(a_i : i \in [n])$.  This gives us weights $(w_i: i \in [n])$, which are essentially a distribution (the weights sum to \(d\) instead of 1).  We remove each point $a_i$ from the dataset with probability proportional to $w_i$.  
%
If significant fraction of the weight $(w_i: i \in [n])$ is supported on the outlier points, then this strategy makes progress in removing many outliers.  We think of this in terms of a tunable parameter \(\gamma > 0\).  If in some iteration, the dual solution places at least \(\gamma d\) (out of the total \(d\)) weight on outlier points, then in expectation, out of the $d$ discarded points at least \(\gamma \gg \alpha \) fraction of it will be outliers. 
Hence in every iteration, as long as the dual program (dual-MVEE) puts \(\ge \gamma d\) mass on the outliers, we will continue making progress by decreasing the fraction of outliers in the remaining points, while throwing out approximately a \(\gamma/\alpha\) factor more points than the target \(\alpha n\).  

\paragraph{Bounding the volume approximation.} Eventually, once we have thrown out many outliers, on some iteration the dual will place less than \(\gamma d\) weight on outlier points.  In this case, the dual solution $(w_i: i \in [n])$ is mostly supported on the inlier points i.e., on the boundary of the optimal solution, most of the points, by weight, are inliers.  The main challenge is to prove that the volume of the enclosing ellipsoid in this iteration is a good approximation of the volume of \(E^\star\).  

\vnote{03/31 added w.r.t. reviewer feedback:} This is particularly tricky, because the SDP that we solve at this iteration is \emph{not a relaxation} of enclosing only the inliers.  We must instead tie the cost of the SDP on this instance to the cost on a different instance that contains only the inliers.  This is also challenging because we do not make {any assumptions} on the distribution of the inliers --- recall that the inliers are simply defined to be the points inside the optimal solution $E^\star$.  This is different from the standard setting in robust estimation, where the algorithm exploits strong properties of the inliers such as isotropicity or even higher-order moment conditions, which can be algorithmically certified.  

A simple procedure allows us to assume that all of the points \(a_i\) are within a bounding ball of radius \(O(\beta)\), and that the minimum axis length of \(E^\star\) is at least \(1\).  

\anote{12/22 added:} We remark that this rescaling is done once at the start of the algorithm before performing the iterative outlier removal procedure. 
Essentially, we can throw away any points that are very far from the majority of points, since $E^\star$ has condition number \(\le \beta\), so it cannot stretch too far from the majority of inliers.

The following lemma is the crucial claim in the analysis and shows that the minimum enclosing ellipsoid in any iteration where the dual places at most $o(d)$ weight on outliers has small volume competitive with the optimal solution $E^\star$ (that encloses only the inliers). 

\begin{lemma}[Good volume approximation when total outlier weight is $o(d)$]\label{lem:overview:volume-approx}
    Let $A$ be a set of points $a_1, \dots, a_n \in \R^d$ with $\norm{a_i} \le \beta$, and $E^\star = \{x: x^\top S^{-1} x \leq 1\}$ be an ellipsoid \anote{12/22 added rescaling:} with $ I \preceq S \preceq O(\beta) \cdot I$.  
    If the optimal dual solution (dual-MVEE) $(w_i: i \in [n])$ has total weight at most \(\gamma d\) on the outliers $A \setminus E^\star$, then the volume of the ellipsoid $E$ corresponding to this dual solution is upper bounded by 
    \[
    \log\Big(\frac{1}{\kappa_d}\vol(E)\Big) = \log\det \Big(\sum_{i=1}^n w_i a_i a_i^\top  \Big) \le \log \det(S) + O(\gamma d \log\beta ) ~ \implies ~ 
    \vol(E) \le \vol(E^\star) \cdot \beta^{O(\gamma d)}.\]
    Here $\kappa_d>0$ is the volume of the unit ball of radius $1$ in $d$ dimensions. 
\end{lemma}
(See \Cref{lem:vol-approx} for the formal version of this lemma in the general non-origin centered setting.) 

To develop some intuition for this lemma, 
imagine that the weights \(w_i\) are supported on exactly \(d\) points i.e.,  \(w_i = 1\) for exactly $d$ out of the $n$ points, and \(w_i = 0\) for all others. Then, if the total weight on the outliers is at most \(\gamma d\) as in the assumption of the lemma, then there are at most $\gamma d$ outlier points 
on the boundary of the enclosing ellipsoid. Any ellipsoid that encloses all of these $d$ points in the support of $w$ also encloses all the points. 
Hence we can upper bound the volume of the minimum enclosing ellipsoid of all of the $n$ points 
is 
by constructing another ellipsoid that contains these \(d\) points.  First, we contain the $(1-\gamma)d$ inlier points by using \(E^\star\).  Then, there are at most \(\gamma d\) remaining outlier points that need to be included.  For each outlier point, we will stretch \(E^\star\) in the direction of the outlier point to include it.  Since \(E^\star\) has width at least \(1\) in every direction, and each outlier point is at distance \(O(\beta)\) away, we expand the volume of \(E^\star\) by a factor at most \(O(\beta)\) for each outlier point that we stretch to include.  Thus, in this ideal situation when the weights $w$ from the optimal dual solution are supported on exactly $d$ points, with total weight of $\gamma d$ on the outliers, the volume of the enclosing ellipsoid is at most  
\[\vol(E) \le \vol(E^\star) \cdot O(\beta^{\gamma d}).\]

\paragraph{Challenge: spread outliers.} Of course, in general, the dual weights could be supported on many more than \(d\) points.  This presents a challenge for the type of analysis sketched above, as that uses not only that there is at most \(\gamma d\) weight on outlier points, but also {\em crucially that it is concentrated in at most \(\gamma d\) directions}.  When the outlier weight could be spread out in many directions, it can affect the value of the dual solution significantly more.  


Consider an example where \(E^\star\) has width \(\beta\) in the first \(d/2\) coordinate directions, and width \(1\) in the remaining \(d/2\) coordinate directions.  Then, the volume of \(E^\star\) is proportional to \(\beta^{d/2}\).
Suppose the outliers are spread out at distance \(\beta\) in the ``thin" directions of \(E^\star\).  
\anote{We used $M$ here. Please check. }
Then, even if the outliers have total weight at most \(\gamma d\), \(M = \sum_{i} w_i a_i a_i^\top\) could look like 
\[
M= \begin{bmatrix}
\begin{array}{c|c}
2(1-\gamma) \beta^2 I_{d/2} & 0 \\ \hline
0 & 2 \gamma \beta^2 I_{d/2}
\end{array}
\end{bmatrix}, \text{ with } \det(M) = \beta^d 2^{d/2} (1 - \gamma)^{d/4}\gamma^{d/4} = \beta^{d/2} \cdot (4 \beta^2 (1 - \gamma)\gamma)^{d/4}.
\]

\anote{Draw a figure here?}
Hence, its volume of the corresponding ellipsoid could be larger than $\vol(E^\star) = \beta^{d/2}$ by a factor of  \((\beta^2 \gamma)^{d/4 }\approx \beta^{\Omega(d)}\).  Thinking of \(\beta\) as a large constant, or even polynomially large in \(d\), we would like the exponent in our approximation factor to be brought down by \(\gamma\) as it was in the previous example with only \(d\) points.

\paragraph{Using optimality via complementary slackness.} The key claim that allows us to prove \Cref{lem:overview:volume-approx} is that this bad situation when the \(\gamma d\) total weight on outliers is spread in \(\omega( \gamma d)\) directions, will actually never happen in an \emph{optimal} dual solution!  Specifically, in the above example, the dual would actually be able to achieve a much larger determinant by putting \(\approx d/2\) weight on the inliers and \(\approx d/2\) weight on the outliers.  This is good, because then we will make progress in removing outliers.  

More generally, we will use the \emph{complementary slackness} or KKT conditions for the optimality of the solutions to (MVEE-OC) and (dual-MVEE). Moreover, we will use the following important geometric inequality that uses the log-concavity of the determinant function. This is a geometric form of the famous {\em Brascamp-Lieb} inequality due to Keith Ball that was used to prove certain reverse isoperimetric inequalities~\cite{Ball1989VolumesOS, Ball1991Volume}. 

\begin{restatable}[Geometric Brascamp-Lieb inequality]{claim}{geometricbrascamplieb}\label{clm:logconcave}
    For any PSD matrix $P \in \mathbb{R}^{d \times d}$, weights $w_i > 0$, and unit vectors $\|b_i\| = 1$ such that $\sum_{i} w_i b_i b_i^\top = I_{d}$, we have
    $$
    \log \det (P) \leq \sum_i w_i \log(b_i^\top P b_i).
    $$
\end{restatable}

\anote{Need to check if $Q$ is used consistently. }
The geometric interpretation implication of the claim is the following: the volume of any ellipsoid can be upper bounded by an appropriate geometric mean of the projected lengths of the ellipsoid onto a set of unit vectors that are isotropic (or more generally form a ``tight frame'').  
We now present a short proof of \Cref{lem:overview:volume-approx} assuming \Cref{clm:logconcave}. In \Cref{sec:algorithm} we 
show how it yields the general version in \Cref{lem:vol-approx} that is needed for the analysis of the general algorithm.

\paragraph{Proof of \Cref{lem:overview:volume-approx} for the Origin-centered Setting.}
Let $w$ be the optimal solution to the dual program. Let $M = \sum_i w_i a_i a_i^\top$.
Let $E^{\star} = \{x: x^\top S^{-1} x \leq 1\}$ be the optimal ellipsoid that encloses the inliers. Then, the volume of $E^{\star}$ is proportional to $\det(S)$.

To compare the volumes of $E^\star$ and $E$ consider the PSD matrix $P = M^{1/2}S^{-1} M^{1/2}$. Note that $S$ is full rank by assumption. We have 
$$
\log \det (P) = \log \det (M) + \log\det(S^{-1}).
$$

Consider all points $a_i$ with weight $w_i > 0$ in the optimal dual solution. By the KKT condition (complementary slackness), we have for any such point $a_i$ with $w_i > 0$, $a_i^\top M^{-1} a_i = 1$. Let $b_i \coloneqq M^{-1/2} a_i$. Then this implies that for all $i$ such that $w_i>0$, the transformed points have length $\|b_i\|_2 = 1$ (i.e., they are on the boundary) and 
$$
\sum_{i: w_i >0} w_i b_i b_i^\top = M^{-1/2}\left(\sum_{i: w_i >0} w_i a_i a_i^\top\right) M^{-1/2} = I.
$$
Hence the subset of points $b_1, \dots, b_n$ with weights $w_i>0$ are isotropic i.e., they form a tight frame.  By \Cref{clm:logconcave}, we have 
$$
\log\det(P) \leq \sum_{i:w_i >0} w_i \log(b_i^\top P b_i) = \sum_{i:w_i >0} w_i \log(a_i^\top M^{-1/2} P M^{-1/2} a_i)= \sum_{i:w_i >0} w_i \log(a_i^\top S^{-1} a_i),
$$
since $S^{-1}=M^{-1/2} P M^{-1/2}$. 

We now split the right-hand side into inliers and outliers. For the inliers $a_i \in E^{\star}$, we have $a_i^\top S^{-1} a_i \leq 1$. 
For the outliers, we have \(\|a_i \| \le  \beta\) (from the assumptions of the lemma), and \(E^\star\) has condition number \(O(\beta)\).  
Hence $a_i^\top S^{-1} a_i \leq O(\beta)$ for all outliers. We are given that the total weight $w$ on the outliers is at most $\gamma d$. Then, we have 
$$
\log \det(P) \leq \sum_{i:w_i >0} w_i \log(a_i^\top S^{-1} a_i) \leq \gamma d \cdot O(\log \beta). 
$$
Therefore, we bound the optimal dual solution value by
$$
\log \det (M) = \log \det(S) + \log \det(P) \leq \log \det(S) + \gamma d \cdot O(\log \beta).
$$
This completes the proof of \Cref{lem:overview:volume-approx}. 

We remark that the above proof does not use optimality of the ellipsoid $E^\star$ for the inliers. In particular, the above proof uses complementary slackness conditions for $M$, and not for $S$. Hence $E^\star$ can be any ellipsoid of condition number $\le \beta$ containing $(1-\alpha)n$ points, and our algorithm is guaranteed to find an ellipsoid whose volume is competitive with it. In fact, the algorithm does not need the knowledge of $\beta$, and outputs an ellipsoid for which the guarantee holds for any choice of $\beta$. \anote{12/18: added this line.}

We defer the proof for the general (non-origin-centered) setting, and analysis of the entire algorithm to \Cref{sec:algorithm}.

\subsubsection{Robust Subspace Recovery}
We present an algorithm for robust subspace recovery by leveraging our algorithm for finding the minimum volume confidence ellipsoid. The approach builds on the geometric intuition that if many points lie near a low-dimensional subspace, the smallest-volume ellipsoid covering them must be elongated along that subspace. We then apply a PCA-style procedure that selects the eigenspace corresponding to the large axes lengths of the ellipsoid as an approximation to the underlying subspace.  This is related to the example that we discussed above where the outliers are spread out.  Due to log-concavity of volume, the volume of the ellipsoid is smaller if we stretch the ellipsoid a lot but only in a few directions, as opposed to stretching it a little bit in many directions.  

The main challenge is that the ellipsoid produced by our algorithm may not be well-conditioned. In some directions, it can have very small eigenvalues, which makes it difficult to bound the dimension of the recovered subspace.
To address this, we introduce a random perturbation step. The algorithm adds small, independent Gaussian noise to each point, which enforces an anti-concentration property. This property ensures that in every direction, only a small fraction of points have very small projections. As a result, any ellipsoid that covers most points must be sufficiently “fat’’ in all directions (see \Cref{lem:fatness:random} for the technical claim). 
This ensures a lower bound on the smallest eigenvalues, allowing us to recover a low-dimensional subspace from the ellipsoid’s principal directions.

\subsubsection{Hardness}
We establish the Small Set Expansion (SSE) hardness for our problem, using the robust subspace recovery problem as an intermediate step. 
Hardt and Moitra proved the SSE-hardness of robust subspace recovery for low-dimensional subspaces ($\delta d$-dimensional), which shows that distinguishing whether there exists a low-dimensional subspace containing a certain fraction of points is SSE-hard~\cite{HardtMoitra13}. 
We then connect this to our algorithmic result for subspace recovery in Section~\ref{sec:subspacerecovery}, which shows that if there exists a small-volume ellipsoid covering most of the points, then those points must lie close to some subspace. 

There are two main challenges. (1) The first challenge is that our PCA-based argument introduces a small gap in the subspace dimension, while the hardness by~\cite{HardtMoitra13} assumes the same dimension in both YES and NO cases. To bridge this, we modify their reduction to tolerate such a dimensional slack, yielding a strengthened hardness for robust subspace recovery with a small dimension gap. 
(2) Finally, after establishing that many points lie close to a low-dimensional subspace, we still need to argue that there exists a subspace that exactly contains a large fraction of these points.
To address this, we replace the randomized construction used by~\cite{HardtMoitra13} with a deterministic construction of the robust subspace recovery instance. 
This discrete structure ensures that when many points are close to a subspace, we can identify a subspace of comparable dimension that exactly contains a large subset of them (\Cref{lem:exact-contain}), thereby completing the reduction and establishing the SSE-hardness of the minimum volume confidence ellipsoid problem.


\subsection{Related Work}\label{sec:related}
\paragraph{Finding Small Volume Confidence Sets}
The problem of learning minimum volume confidence sets has been widely studied in statistics and machine learning, due to applications in uncertainty quantification, learning density level sets, anomaly detection, or distribution support estimation~\cite{einmahl1992generalized, polonik1997minimum, rousseeuw1985multivariate, garcia2003level, scott2005learning, gao2025volume}. The works of \cite{polonik1999concentration, scott2005learning} design statistically efficient methods when the VC dimension of the class of confidence sets $\mathcal{C}$ is bounded. This problem has also received recent attention due to efficiency or volume optimality considerations in conformal prediction and uncertainty quantification~\cite{lei2013distribution, Sadinle2016LeastAS, gao2025confidence, gao2025volume, braun2025volume, srinivas2025online}; see \cite{gao2025confidence, gao2025volume} for more references.  

The most relevant prior work is the work on small volume learning confidence balls in high dimensions~\cite{gao2025confidence}. The class $\mathcal{C}$ of Euclidean balls in $d$ dimensions is already computationally challenging for large $d$. The work of~\cite{gao2025confidence} studies both improper learning (where the confidence set that is output does not need to be from $\mathcal{C}$) and proper learning algorithms, and designs polynomial time algorithms with rigorous guarantees of approximate volume optimality  
for $\mathcal{C}$ being the set of Euclidean balls (or unions of $O(1)$ balls). The best known polynomial time proper learning algorithm for Euclidean balls achieves $\Gamma=1+\tilde{O}(1/\log d)$ (volume approximation of $\exp(\tilde{O}(d/\log d)$) ~\cite{badoiu2002approximate}, while the best improper learning algorithms for learning confidence balls outputs ellipsoids that achieve $\Gamma=1+\tilde{O}(1/\sqrt{d})$ i.e., a volume approximation of $\exp(\tilde{O}(\sqrt{d}))$ ~\cite{gao2025confidence}. This is complemented by a computational intractability result for proper learning better than $\Gamma=1+d^{-\epsilon}$ for some constant $\epsilon>0$. The main technical idea is to design a robust proxy for computing the center, inspired by algorithms for robust mean estimation. 

These results do not give any guarantees when competing against the class $\mathcal{C}$ of ellipsoids in $d$ dimensions. For the learning confidence ellipsoids, the challenge is to learn both the center (location), and more importantly, the shape of a good ellipsoid for the distribution. It is an interesting open question to generalize our results to unions of ellipsoids, or use our new algorithm to get an improved improper algorithm for learning Euclidean balls. 

\paragraph{Comparison to Prior Work on Robust Estimation.}
A related direction that has received much attention in both statistics and computer science is robust estimation, where the goal is to develop methods for parameter estimation that are robust to outliers. Here, the uncorrupted data points i.e., {\em the inliers} are drawn i.i.d from a nice distribution $\mathcal{D}_{\theta}$ over $\R^d$ with a parameter $\theta$ that we would like to estimate, e.g., a Gaussian with an unknown mean, or unknown covariance. In the input dataset, $1-\alpha$ fraction of the samples are drawn from the inlier distribution, while the remaining $\alpha$ fraction are {\em outliers}, which are arbitrary points in $\R^d$ (stronger contamination models allow for replacing an arbitrary $\alpha$ fraction of the inliers with arbitrary outliers). There is a rich body of work over the past six decades developing robust methods for several basic parameter estimation tasks under different contamination models~\cite[see e.g., extensive books on this topic][]{huber2004robust,ars_book}. Recent work has even used ideas from robust estimation to find small confidence balls~\cite{gao2025confidence}.  

Most relevant to us are the works on robust covariance estimation, and robust estimators for the scatter of the distribution in high dimensions. The important works of \cite{diakonikolas2016robust,lai2016agnostic} 
 gave polynomial time algorithms for robust mean and covariance estimation of Gaussians in high dimensions, with near optimal dependencies.  
 Over the past decade there have been several polynomial time algorithms proposed based on filtering, sophisticated convex relaxations including sum-of-squares and many other techniques~\cite[please see][for a recent book on the topic]{ars_book}. In particular, for robust covariance estimation, these works assume Gaussianity of the inlier distribution, or various fourth moment or higher-order moment assumptions on the inlier distribution. Such assumptions are necessary to define a certifiable property (e.g., hypercontractivity, anticoncentration) of the inlier distribution in parameter estimation tasks.  In contrast, in our setting we do not make any assumptions on the data points --- in particular, the points in the optimal solution do not need to satisfy any Gaussianity or hypercontractivity or anticoncentration condition. 
 Handling arbitrary distributions is important for applications in uncertainty quantification. Moreover, our goal is not to estimate or recover any ground-truth parameter, but instead to find some small volume confidence ellipsoid that contains most points.   

The classic work of Rousseeuw~\cite{rousseeuw1985multivariate} studies the same problem as the famous {\em minimum volume estimator} (MVE) for location and scatter parameters of a distribution. See the survey by \cite{van2009minimum} for several desirable properties including its high breakdown point and other statistical properties.  
As described earlier, this problem and heuristics for it have also found applications in control theory, dynamical systems, and uncertainty quantification. However, known algorithms for this problem are either computationally inefficient, or do not come with rigorous guarantees of (approximate) volume optimality.  

\paragraph{Robust Subspace Recovery.} Finally, the problem of robust subspace recovery in high dimensions has received significant attention on the algorithms side in the past decade starting with the influential work of Hardt and Moitra~\cite{HardtMoitra13}. Polynomial time algorithms that recover a ground-truth subspace in the presence of outliers are only known under additional distributional assumptions about the inliers (points on the subspace)~\cite{bakshi2021list}, or non-degeneracy assumptions/ the smoothed analysis setting on the outliers~\cite{HardtMoitra13, BCPV}. In particular, the work of Hardt and Moitra gives polynomial time algorithms that recover the ground truth subspace of dimension $k$ when an $\alpha < 1-k/d$ fraction of outliers are in general position in the ambient space, and the $1-\alpha$ fraction of inliers are in general position in the subspace. One of their algorithms uses a beautiful connection to radial isotropy~\cite{HardtMoitra13}. This uses a different convex program due to Barthe~\cite{Barthe1997}, which is also related to a reverse version of the Brascamp-Lieb inequality; see also ~\cite{hopkins2020pointlocationactivelearning, Diakonikolas2022ASP} for related work on the Forster transform and the radial isotropy position. Please see \Cref{sec:subspacerecovery} for a more detailed comparison to prior work on robust subspace recovery.

\section{Algorithm and Analysis}\label{sec:algorithm}
\vnote{need to revisit this overview}
Let \(E^\star\) be the smallest ellipsoid that contains at least \((1 - \alpha)\)-fraction of the points in \(A=\{a_1, \dots, a_n\}\). 
We begin by finding a bounding ball \(\widehat{B}\) that contains all points in \(E^\star\), and has radius at most twice that of the maximum axis length of \(E^\star\).  This can be done by a simple search over all pairs of points in \(A\) (\Cref{lem:bounding-ball}).   
Restricted to \(\widehat{B}\), it suffices to find an ellipsoid that captures close to \((1 - \alpha)n\) of the embedded points, with volume approximately that of the smallest volume ellipsoid of condition number at most \(\beta\) that captures \((1 - \alpha)n\) points.   

We will refer to the \((1 - \alpha)n\) points in $A$ that are captured by the optimal \(E^\star\) 
as \emph{inliers}.  We refer to the remaining \(\alpha n\) points in $A$ 
as \emph{outliers}.  

Our algorithm proceeds by iteratively solving the dual program of the minimum volume ellipsoid (\Cref{prop:dual-MVE}).  This gives us weights, which are essentially a distribution supported on the \(a_i\) (the weights sum to \(d\) instead of 1).  We remove each point from the dataset with probability that is proportional to the weight it is assigned in the dual.  For a tunable parameter \(0 < \gamma < 1\), we show the following.  In a given iteration, consider the following cases. 
\begin{enumerate}[(a)]
    \item If more than \(\gamma d\) weight is placed on outliers, then when we remove points, more than a \(\gamma\) fraction of the points that are removed will be outliers, in expectation.  

    \item If less than \(\gamma d\) weight is placed on outliers, then the minimum volume ellipsoid enclosing all of the remaining points is already a good approximation in volume to the minimum volume ellipsoid enclosing only the inliers (\Cref{lem:vol-approx}). 
\end{enumerate}

Since we are only removing points in each iteration, the volume of the minimum volume enclosing ellipsoid of the remaining points can only decrease in each iteration.  Thus if we iterate \(R \approx \frac{\alpha n}{\gamma d}\) times, we will have by the pigeonhole principle that on one of these iterations, less than \(\gamma d\) mass was placed on outliers, and therefore the volume approximation held on that and all following iterations.  The fact that we make progress removing outliers be shown with Chebyshev's inequality, after observing that the point removal process forms a Martingale (\Cref{lem:outlier-removal-progress}). 
We can also ensure that we do not remove too many points, inliers and outliers, over the \(R\) rounds via Chebyshev's inequality (\Cref{lem:dont-remove-too-many-points}).  

Finding the bounding ball, and part (a) of the outline above, which argues about the number of points removed, have proofs relegated to \Cref{app:algo-details}.  In this section, we prove (b), which is the main technical contribution.

\subsection{SDP and KKT condition}

The main technical statement that we need to prove is about the form of the optimum of SDP for the minimum volume enclosing ellipsoid.  We are interested in finding the minimum volume enclosing ellipsoid (MVEE) of \emph{arbitrary center}.  This can be found by the following convex program. 

For a given set of points \(A\), the minimum volume ellipsoid that contains \emph{all} of the points \(A\) can be found via the following convex program.  
\begin{proposition}[Minimum Volume Enclosing Ellipsoid, see e.g.\ \cite{SunF2004, ToddY2007} ] \label{prop:mvee-center}
    For a set of points \(a_1, \dots, a_n\), the minimum volume ellipsoid enclosing the points is given by 
    \[\{x : (x-c_\star)^\top S_\star (x-c_\star) \le 1 \}\]
    for X = \(S_\star^{1/2}\) and $z = S_\star^{1/2} c_\star$ be the optimum of the convex program
    \begin{align*}
        \min_{X \in \mathbb{R}^{d \times d}, z \in \R^d } & - 2\log\det X  \\
        \text{s.t. } & \|Xa_i - z\|_2^2 \le 1, \quad \forall i \in [n],   \\
        & X \succeq 0.
    \end{align*}
\end{proposition}

The dual of this program is the following.
\begin{proposition}[Dual program of MVEE]\label{prop:dual-MVE-center}
    The dual of the program in \Cref{prop:mvee-center} is given by the convex program
    \begin{align*}
        \max_{w_i} ~&\log \det \left( \sum_{i} w_i (a_i-c) (a_i-c)^\top \right) \\
        \text{s.t. } ~& w_i \ge 0, \quad \forall i \in [n] \\
        &\sum_{i} w_i = d\\
        & c = \frac{1}{d}\sum_i w_i a_i
    \end{align*}
    We refer to the value \(w_i\) as the \emph{weight} assigned to \(a_i\).  While it is not given as an explicit constraint, it is without loss of generality that in the optimal solution, every \(w_i \le 1\).  
    Furthermore, strong duality holds, and the optimal value of the primal program is equal to the optimal value of the dual program.  
\end{proposition}

Another way to find the MVEE (minimum volume enclosing ellipsoid of arbitrary center) is to reduce to the problem of finding the minimum volume origin-centric ellipsoid.  That is, for every point \(a_i \in \mathbb{R}^{d}\), we identify it with the ``lifted" point 
\[\widetilde{a}_i \in \mathbb{R}^{d + 1} = [a_i, ~1].\]
Let \(E^\star\) be the MVEE of the original \(a_i\)s.  Let \(\widetilde{E}\) be the minimum volume origin centric ellipsoid containing the \(\widetilde{a}_i\)s.  These have the nice property that \(E^\star\) is exactly the \emph{slice} of \(\widetilde{E}\) corresponding to the last index being \(1\),
\[E^\star = \widetilde{E} ~\cap~ \{x \in \mathbb{R}^{d + 1} : x_{d + 1} = 1\} .\]
The minimum volume origin centric ellipsoid and its dual problem are given by the following SDPs.  We note that the optimal dual solution for the MVEE problem on the \(a_i\)s (\Cref{prop:dual-MVE-center}), and the optimal dual solution for the origin centric minimum volume origin centric ellipsoid containing the \(\widetilde{a}_i\)s are exactly the same.  We provide a proof in \Cref{prop:dual-equiv-existence}.  To prove the volume approximation guarantee, we find it convenient to adopt the viewpoint of the origin centric dual.  

\begin{proposition}[Minimum Volume Origin Centric Ellipsoid, see e.g.\ \cite{boyd2004convex}] \label{prop:centered-mve}
    For a set of points \(a_1, \dots, a_n\), the minimum volume centered ellipsoid enclosing the points is given by 
    \[\{x : x^\top (M^\star)^{-1} x \le 1 \}\]
    for \((M^\star)^{-1} = Q\) the optimum of the SDP
    \begin{align*}
        \max_{Q \in \mathbb{R}^{n \times n}} &\log\det Q  \\
        \text{s.t. } & \forall i \in [n], ~ a_i^\top Q a_i \le 1  \\
        & Q \succeq 0.
    \end{align*}
\end{proposition}

\begin{proposition}[Dual program of origin centered MVE, see e.g.\ \cite{boyd2004convex}]\label{prop:dual-MVE}
    The dual of the program in \Cref{prop:centered-mve} is given by the SDP
    \begin{align*}
        \max_{w_i} ~&\log \det \left( \sum_{i} w_i a_i a_i^\top \right) \\
        \text{s.t. } ~& \forall i \in [n], ~w_i \ge 0 \\
        &\sum_{i} w_i \le d.
    \end{align*}
    We refer to the value \(w_i\) as the \emph{weight} assigned to \(a_i\).  While it is not given as an explicit constraint, it is without loss of generality that in the optimal solution, every \(w_i \le 1\).  
    Furthermore, strong duality holds, and the optimal value of the primal program is equal to the optimal value of the dual program.  
\end{proposition}

\begin{proof}[Proof of~\Cref{prop:dual-MVE-center}]
Let $w_i\ge 0$ be Lagrange multipliers for the constraints
$\|Xa_i-z\|_2^2\le 1$. The Lagrangian relaxation of MVEE is
$$
\mathcal{L}(X,z,w)
= -2\log\det X + \sum_{i=1}^n w_i (\|Xa_i-z\|_2^2-1 ).
$$
\emph{Stationarity in $z$.} Taking the gradient in $z$ and setting it to zero gives
$$
0=\nabla_z\mathcal{L}=2\sum_i w_i(z-Xa_i)
\ \Longrightarrow\
z=X\,c,\qquad
c:=\frac{\sum_i w_i a_i}{\sum_i w_i}.
$$
Define the matrix
$$
M:=\sum_{i=1}^n w_i(a_i-c)(a_i-c)^\top \succeq 0.
$$
With $z=Xc$, the quadratic term in the Lagrangian is
$$
\sum_{i=1}^n w_i\|Xa_i-z\|_2^2
=\sum_{i=1}^n w_i\|X(a_i-c)\|_2^2
=\operatorname{tr}(X^\top X M).
$$
Let $S=X^2$. Using
$-2\log\det X=-\log\det S$, the Lagrangian becomes
\[
\mathcal{L}(S,w)=-\log\det S+\operatorname{tr}(SM)-\sum_i w_i.
\]
\emph{Stationarity in $S$.} Minimizing over $S\succ 0$ yields the unique optimizer
$S=M^{-1}$, and the minimum value
\[
\min_{S\succ 0}\mathcal{L}(S,w)
=\log\det M + d-\sum_i w_i.
\]
Hence, the dual program of MVEE is
\begin{align*}
\max_{w\ge 0}\ &\log\det M +d-\sum_i w_i \\
\text{s.t. } & c=\frac{\sum_i w_i a_i}{\sum_i w_i}, \\
& M = \sum_{i=1}^n w_i(a_i-c)(a_i-c)^\top.
\end{align*}
Scaling all weights by a positive factor $t$ leaves $c$ unchanged and scales
$M$ by $t$. Thus, we have $\log\det M =\log\det M +d\log t$ while
$-\sum_i tw_i=-t\sum_i w_i$. Optimizing over $t$ makes the total weight
$\sum_i w_i = d$ at optimum.
Thus, we get the dual program.

\emph{Strong duality.} The primal program MVEE has a strictly feasible solution. Therefore, strong duality holds and the optimal values of the
primal and dual are equal.

\emph{Bounds $w_i\le 1$ at optimum.} 
Let $(X^\star, z^\star)$ be any optimal primal solution. Let $w_i^\star$ be the optimal dual solution and $M^\star$ be the corresponding matrix of $w_i^\star$. Then, we have $(M^\star)^{-1} = (X^\star)^2$, which implies
$$
I_d  = X^\star M^\star X^\star \;=\; \sum_i w_i (X^\star(a_i-c))(X^\star(a_i-c))^\top.
$$
Let $y_i=X^\star(a_i-c)$. By complementary slackness (\Cref{prop:cs}), if $w^\star_i>0$ then
$\|y_i\|=1$. Thus, we have 
$$
1= y_i^\top I_d y_i = \sum_j w_j (y_i^\top y_j)^2 \ge w_i (y_i^\top y_i)^2 = w_i,
$$
so $w_i\le 1$ for all $i$.
\end{proof}

\begin{proposition}[Complementary Slackness]\label{prop:cs}
Let $(X^\star, z^\star)$ be any optimal primal solution to the MVEE problem in
\Cref{prop:mvee-center}, and let $w^\star=(w_1^\star,\dots,w_n^\star)$ be any
optimal dual solution to \Cref{prop:dual-MVE-center}. Then, for every $i\in[n]$,
$$
w_i^\star\big(\,\|X^\star a_i - z^\star\|_2^2 - 1\,\big)=0.
$$
In particular, if $w_i^\star>0$ then $\|X^\star a_i - z^\star\|_2^2 = 1$, and if
$\|X^\star a_i - z^\star\|_2^2 < 1$ then $w_i^\star=0$.
\end{proposition}

\begin{proof}
Let $(X^\star,z^\star)$ be any optimal primal solution and $w^\star$ be any optimal dual solution. The KKT complementary slackness
conditions are
$$
w_i^\star\big(\|X^\star a_i-z^\star\|_2^2 - 1\big)=0,\qquad i=1,\dots,n.
$$
Hence, for any index with $w^\star_i>0$, we have $\|X^\star a_i-z^\star\|_2^2=1$.
\end{proof}

\subsection{Volume Approximation}

The main component of our proof is a statement that if the optimal dual solution to the origin centric minimum volume enclosing ellipsoid problem does not put much weight on outlier points, then the volume of the enclosing ellipsoid of all of the points is comparable to the smallest ellipsoid of bounded condition number, that contains all of the inliers.  Thus, the outliers that are still included do not blow up the volume of the enclosing ellipsoid by too much.  

The crux of the argument is in the following claim. 

\anote{12/18: changed $Q$ to $P$ to avoid confusion.}
\vnote{use thm-restate with 1.6}
\geometricbrascamplieb*

\begin{proof}
    Let $P = U\Lambda U^\top$ be the SVD decomposition of $P$, where $\Lambda=\diag(\lambda_1,\cdots,\lambda_{d + 1})$. Let $c_i = U^\top b_i$. Then, we have $\|c_i\| = \|b_i\| =  1$ and $\sum_i w_i c_ic_i^\top = \sum_i w_i b_ib_i^\top = I_{d+1}$. For each $i$, we have 
    $$
    \log(b_i^\top P b_i) = \log(c_i^\top \Lambda c_i) = \log \left(\sum_{j=1}^{d+1} \lambda_j c_{ij}^2 \right) \geq \sum_{j=1}^{d+1} c_{ij}^2\log(\lambda_j),
    $$
    where the last inequality is due to the concavity of the log function and Jensen's inequality.

    Since $\sum_i w_i c_ic_i^\top = I_{d+1}$, for every $j$, we have $\sum_i w_i c_{ij}^2 = 1$. Thus, we have 
    $$
    \sum_i w_i \log(b_i^\top P b_i) \geq \sum_i w_i \sum_{j=1}^{d+1} c_{ij}^2\log(\lambda_j) = \sum_{j=1}^{d+1} \log(\lambda_j) = \log\det(P). 
    $$
\end{proof}

We now bound the volume of the ellipsoid provided by our algorithm. 
We embed each point $a_i \in \R^d$ to point $\wta_i := (a_i,1) \in \R^{d+1}$ by adding the $(d+1)$th coordinate with value $1$. 
We then solve the origin-centered minimum volume ellipsoid for $\wta_1,\dots, \wta_n \in \R^{d+1}$.
Let $w$ be the optimal solution to the corresponding dual program in~\Cref{prop:dual-MVE}. Then, we consider the ellipsoid $E$ in $\R^d$ as follows. Let center $c = \sum_i w_i a_i/(d+1)$ and $M = \sum_i w_i (a_i - c)(a_i-c)^\top$. The ellipsoid $E$ is $\{x: (x-c)^\top M^{-1} (x-c) \leq \frac{d}{d+1}\}$.

\begin{lemma}[Volume approximation]\label{lem:vol-approx}
    Let 
    \(a_i, \dots, a_n \in \mathbb{R}^{d}\) be points contained in a ball of radius \(2 \beta\), and  
    $E^\star = \{x: (x-c_\star)^\top (S^\star)^{-1} (x-c_\star) \leq 1\}$ 
    be an ellipsoid 
    with width at least \(1\) in every direction
    that encloses the inliers.  If the optimal dual solution (\Cref{prop:dual-MVE}) places at most \(\gamma (d+1)\) weight on outliers, then we have the shape matrix $\overline{M}=\frac{d}{d+1}M$ of the ellipsoid $E = \{x: (x - c)^\top (\overline{M})^{-1} (x - c) \le 1 \}$ satisfies
    \[\log \det (\overline{M}) \le \log \det(S^\star) + \gamma (d+1) \cdot \log (16\beta^2).\]
\end{lemma}

\begin{proof}
    It is convenient to rescale the matrix \((S^\star)^{-1}\) to a matrix \(S^{-1}\) such that 
    \[E^\star = \left\{x: (x-c_\star)^\top S^{-1} (x-c_\star) \leq \frac{d}{d+1} \right\}.\]

    We embed each point $a_i \in \R^d$ to point $\wta_i := (a_i~; 1) \in \R^{d+1}$.  Let $w$ be the optimal solution to the origin centric dual program (\Cref{prop:dual-MVE}) for the lifted instance in $\R^{d+1}$. Then, we have $\sum_i w_i = d+1$.  and 
    the lifted matrix $\widehat M \in \R^{(d+1)\times(d+1)}$
    \begin{align*}
        \widehat M := \sum_{i=1}^n w_i\,\wta_i \wta_i^\top = \begin{pmatrix} \sum_i w_i a_i a_i^\top & \sum_i w_i a_i\\ \sum_i w_i a_i^\top & \sum_i w_i \end{pmatrix}.
    \end{align*}
    
    Let 
    $c = \sum_{i=1}^n w_i a_i /(d+1)$ and $M = \sum_{i=1}^n w_i (a_i - c)(a_i-c)^\top$. 
    We first show that for each lifted point $\wta_i \in \R^{d+1}$ with $\wta_i^\top \widehat M^{-1} \wta_i \leq 1$, its corresponding point $a_i \in \R^d$ is contained in the ellipsoid $E = \{x: (x-c)^\top M^{-1} (x-c) \leq \frac{d}{d+1}\}$.  Note that \(M^{-1}\) is scaled so that the right-hand-side of the inequality is \(\frac{d}{d + 1}\), making it directly comparable to \(S^{-1}\).
    
    To show this, consider the Schur complement $\widehat M/\widehat{M}_{d+1,d+1}$ (\Cref{fact:schur})  of the $(d+1, d+1)$ entry of $\widehat M$, which is the $d\times d$ matrix 
    $$
\widehat M/\widehat{M}_{d+1,d+1} = \sum_{i=1}^n w_i a_i a_i^\top - (d+1)c\cdot \frac{1}{d+1} \cdot (d+1)c^\top = \sum_{i=1}^n w_i (a_i-c)(a_i-c)^\top = M.
    $$
    Then, the inverse of $\widehat M$ is 
    $$
    \widehat M^{-1} =         
        \begin{pmatrix}
        M^{-1} & -M^{-1} c\\
        - c^\top M^{-1} & c^\top M^{-1} c + \frac{1}{d+1}
        \end{pmatrix}.
    $$
    Hence, we have for every point $\wta_i$
    $$
    \wta_i \widehat M^{-1} \wta_i = (a_i - c)^\top M^{-1} (a_i -c) + \frac{1}{d+1},
    $$
    which implies that $a_i \in E$ for every $\wta_i$ with $\wta_i \widehat M^{-1} \wta_i \leq 1$.
    
    We now show that the log determinant of $\widehat M$ is close to that of $M$. By the Schur complement, we have 
    \begin{align*}
        \log\det \widehat M &= \log (d+1) + \log\det\left(\sum_{i=1}^n w_i a_i a_i^\top - \frac{1}{d+1} \cdot \left(\sum_{i=1}^n w_i a_i\right)\left(\sum_{i=1}^n w_i a_i^\top \right)  \right) \\
        &= \log (d+1) + \log\det \left(\sum_{i=1}^n w_i (a_i - c)(a_i - c)^\top\right) = \log (d+1) + \log \det M. 
    \end{align*}

    Recall $E^\star = \{x : (x-c_\star)^\top S^{-1} (x-c_\star) \leq \frac{d}{d+1}\}$ is the ellipsoid 
    with width at least \(1\) in every direction
    that encloses the inliers. Then, the volume of $E^*$ is proportional to $\det(S)$. We define the following lifted matrix $\widehat S^{-1} \in \R^{(d+1)\times (d+1)}$
    \begin{align*}
        \widehat S^{-1} := 
        \begin{pmatrix}
        S^{-1} & -S^{-1} c_\star\\
        - c_\star^\top S^{-1} & c_\star^\top S^{-1} c_\star + \frac{1}{d+1}
        \end{pmatrix}.
    \end{align*}
    This lifted matrix corresponds to the shape of an origin-centered ellipsoid in $\R^{d+1}$. For each embedded point $\wta_i \in \R^{d+1}$, we have 
    \begin{align*}
        \wta_i^\top \widehat S^{-1} \wta_i = (a_i-c_\star)^\top S^{-1} (a_i-c_\star) +\frac{1}{d+1}.
    \end{align*}
    We then show that the log determinant of $\widehat S$ is close to that of $S$. By the Schur complement, we have 
    \begin{align*}
        \log\det (\widehat S^{-1}) = -\log (d+1) + \log\det (S^{-1}) =  -\log (d+1) - \log\det S.
    \end{align*}

    Let $P = \widehat M^{1/2} \widehat S^{-1} \widehat M^{1/2}$. We have 
    \begin{align*}
        \log \det P = \log\det \widehat M + \log\det (\widehat S^{-1}) = \log\det M-\log\det S.
    \end{align*}
    
    We then use the geometric Brascamp-Lieb inequality in \Cref{clm:logconcave} to bound $\log\det P$. 
    Consider all points $\wta_i$ with weight $w_i >0$ in the optimal dual solution.
    By the complementary slackness, we have $\wta_i^\top \widehat{M}^{-1} \wta_i =1$.
    For any such point $\wta_i \in \R^{d+1}$, let $b_i = \widehat M^{-1/2} \wta_i \in \R^{d+1}$. Then, we have $\|b_i\|_2 = 1$ and 
    $$
    \sum_{i: w_i >0} w_i b_i b_i^\top = \widehat M^{-1/2}\left(\sum_{i: w_i >0} w_i \wta_i \wta_i^\top\right) \widehat M^{-1/2} = I_{d+1}.
    $$
    
    Thus, by \Cref{clm:logconcave}, we have 
    \begin{align*}
        \log\det P &\leq \sum_i w_i \log(b_i^\top P b_i) 
        = \sum_i w_i \log \left(\wta_i^\top \widehat S^{-1} \wta_i\right),
    \end{align*}
    since $P = \widehat M^{1/2} S^{-1} \widehat M^{1/2}$ and $b_i = \widehat M^{-1/2} \wta_i$.

    We now split the right-hand side into inliers and outliers. For the inliers $a_i \in E^*$, we have 
    $$
    \wta_i^\top \widehat S^{-1} \wta_i =  (a_i-c_\star)^\top S^{-1} (a_i-c_\star) +\frac{1}{d+1} \leq \frac{d}{d+1}+\frac{1}{d+1} = 1.
    $$ 
    The outliers $a_i \not\in E^*$ are contained in a ball of radius \(2 \beta\), so \(\|a_i - c_\star\| \le 4 \beta\).  We also have that the width of \(E^\star\) is at least \(1\) in every direction, so \(S^{-1} \preceq I\), so 
    \begin{align*}
    \wta_i^\top \widehat S^{-1} \wta_i &= (a_i - c_\star)^\top S^{-1} (a_i- c_\star) \cdot \frac{d}{d + 1} + \frac{1}{d+1}\\
    &\leq 16\beta^2 \cdot \frac{d}{d+1} +\frac{1}{d+1} \leq 16\beta^2.
    \end{align*}
    \vnote{modified above calculation to account for the aspect ratio correctly.  Please double check!}
    Suppose the weight $w$ puts $\gamma (d+1)$ mass on outliers. Then, we have 
    \begin{align*}
    \log \det(P) &\leq \sum_i w_i \log \left(\wta_i^\top \widehat S^{-1} \wta_i\right) \leq \gamma (d+1) \cdot \log\left(16\beta^2\right). 
    \end{align*}
    Therefore, we bound the optimal dual solution value by
    $$
    \log \det (M) = \log \det(S) + \log \det(P) \leq \log \det(S) + \gamma (d+1) \cdot \log (16\beta^2).
    $$

   \noindent Finally recall that, for convenience, we worked with \(M^{-1}\) being a scaling of \(\overline{M}^{-1}\) such that 
    \[E = \left\{x: (x-c)^\top M^{-1} (x-c) \leq \frac{d}{d+1}\right\} = \left\{x: (x-c)^\top \overline{M}^{-1} (x-c) \leq 1 \right\},\]
    and similarly \(S^{-1}\) is a scaling of \((S^\star)^{-1}\) such that 
    \[E^\star = \left\{x: (x-c_\star)^\top S^{-1} (x-c_\star) \leq \frac{d}{d+1}\right\} = \left\{x: (x-c_\star)^\top (S^\star)^{-1} (x-c_\star) \leq 1 \right\}.\]
    Thus we have 
    $$
    \log \det (\overline{M}) \leq \log \det(S^\star) + \gamma (d+1) \cdot \log (16\beta^2).
    $$
\end{proof}

\subsection{Algorithmic guarantee}

\begin{figure}[ht]
\begin{tcolorbox}
\begin{center}
    \textbf{Approximate Coverage Ellipsoid}
\end{center}
\textbf{Input} sample points $A=\{a_1,\dots,a_n\}\subseteq\R^d$, target miscoverage $0<\alpha<1$, 
parameter $0<\gamma<1$ \\
\textbf{Output} an ellipsoid $\widehat{E}\subseteq\R^d$
\begin{enumerate}
\item Initialize the candidate set $\mathfrak{E}\gets\varnothing$.
\item For each ordered pair $(k_1,k_2)\in[n]\times[n]$: \label{algline:main-step}
  \begin{enumerate}
  \item Set a coarse bounding ball $\widehat{B}=B(\widehat{c},\widehat{r})$ with $\widehat{c}=a_{k_1}$ and $\widehat{r}=\|a_{k_1}-a_{k_2}\|$.
  \item Form the lifted set
  \[
  \widetilde{A}\ \gets\ \Big\{\wta_i\in\R^{d+1}:~a_i\in A\cap \widehat{B}\Big\},
  \qquad
  \wta_i\ :=\ \Big[a_i,~ 1\Big].
  \]
  \item Let
  $J = \lceil c \cdot (\alpha n/\gamma d)\rceil$ for a constant $c$.
  Repeat for $j=1,\dots,J$:
    \begin{enumerate}
    \item Solve the \textbf{dual} SDP on $\widetilde{A}$ (Prop.~\ref{prop:dual-MVE}) to obtain weights $\{w_i\in[0,1]\}$ for $\wta_i\in\widetilde{A}$.
    \item Independently remove each $\wta_i\in\widetilde{A}$ with probability $w_i$.
    \end{enumerate}
  \item Let $\widehat{E}$ be the minimum volume enclosing ellipsoid of $\widehat{A}\gets\{a_i\in\R^d:\ \wta_i\in\widetilde{A}\}$.
  \item If $|\widehat{E}\cap A|\ \ge\ (1-O(\alpha/\gamma))\,n$, add $\widehat{E}$ to $\mathfrak{E}$.
  \end{enumerate}
\item Repeat step \ref{algline:main-step} \(O(\log n)\) times.
\item Return the smallest-volume ellipsoid in $\mathfrak{E}$.
\end{enumerate}
\end{tcolorbox}
\caption{Approximation algorithm for the minimum volume ellipsoid covering at least $1-\alpha$ fraction of points.}
\label{fig:ellipsoid_tcolor}
\end{figure}

\anote{12/18: Edited running time. Check!}
\vnote{\(\beta\) does not have to be an input to the algorithm}
\vnote{fix for success probability \(1 - o(1)\)}
\begin{theorem}[Approximately Optimal Confidence Ellipsoid]\label{thm:approx-optimal-confidence-ellipsoid}
   Fix a set of points \(A \subseteq \mathbb{R}^d\), a target miscoverage rate \(0 < \alpha < 1\), a condition number \(\beta \ge 1\), a parameter \(0 < \gamma < 1\), 
   and a success rate \(0 \le c_1 < 1\).  Then, in \(\poly(n, d) \cdot T_{\mathrm{SDP}} \) time, 
   with probability \(1 - o(1)\), the algorithm in \Cref{fig:ellipsoid_tcolor} outputs an ellipsoid \(\widehat{E}\), such that 
    \(\widehat{E}\) approximately satisfies coverage,  
    \[|\widehat{E} \cap A | \ge \left(1 - O\left(\frac{\alpha}{\gamma} \right) \right)n, \]
    and is approximately volume optimal 
    \[\vol(\widehat{E}) \le \vol(E^\star) \cdot ( 4 \beta )^{2 \gamma (d + 1)},\]
    where \(E^\star\) is the minimum volume ellipsoid of condition number at most \(\beta\) that encloses at least \((1 - \alpha)|A|\) points of \(A\). Here $T_{\mathrm{SDP}}$ is the running time for solving the $\poly(n,d)$ sized semidefinite programs in \Cref{prop:centered-mve} and \Cref{prop:dual-MVE}.
\end{theorem}
\anote{1/2: Added}
\vnote{moved into proof: We remark that the lower bound condition on the number of samples $n \ge \gamma d/\alpha$ is essentially without loss of generality: we can take $\ell= \lceil \frac{\gamma d}{\alpha n}\rceil$ copies of each point and run the algorithm on the dataset with $\ell$ copies of $A$. Note that the ellipsoid either covers all copies of a point, or none of it -- hence the coverage and volume guarantees carry over.  }

The running time $T_{\mathrm{SDP}}$ of standard semidefinite programming solvers like the ellipsoid method for $\poly(n,d)$-sized SDPs including \Cref{prop:centered-mve} and \Cref{prop:dual-MVE} is polynomial in $n,d$ and $\log(1/\eta)$ where $\eta>0$ is the desired precision for solving the SDP~\cite{Toddbook}. Algorithms for solving the specific semidefinite program~\Cref{prop:centered-mve} for the minimum volume enclosing ellipsoid has received particular attention, with faster algorithms based on interior-point methods and first-order methods; see \cite{Toddbook, cohen19a} for details and more references. We will ignore numerical precision issues in the analysis, and assume that the SDP can be solved to optimality.     
\begin{proof}
    We can begin by assuming the following parameter settings. 
    If \(\alpha < \frac{1}{n}\), we can simply output the minimum volume enclosing ellipsoid containing all of \(A\), by solving the relevant SDP (\Cref{prop:centered-mve}, with lifted points).  If \(\gamma \le \alpha\), then \(\alpha / \gamma \ge 1\), so a trivial ellipsoid with \(0\) volume and \(0\) coverage satisfies the guarantee.  Thus, we can assume that 
    \(\gamma > \alpha \ge \frac{1}{n}.\)
    This allows us to assume that 
    \begin{equation}\alpha n / \gamma d \ge 1.\label{eq:alpha-n-big}\end{equation} 
    If this is not the case, we can simply include $\ell= \lceil \frac{\gamma d}{\alpha n}\rceil$ copies of each point in \(A\).  Since the ellipsoid either covers all copies of a point or none of them, the coverage and volume guarantees will carry over.  The new number of points 
    \[n' = \ell \cdot n \le \left(\frac{\gamma d}{\alpha n} + 1 \right) n \le \gamma d n + n = \mathrm{poly}(n),\]
    is still polynomial in the original \(n\), and thus the runtime guarantees will still hold.  

    We begin by using the procedure in \Cref{lem:bounding-ball} to construct a set \(\mathfrak{B}\) of at most \(n^2\) potential bounding balls for the points in \(A\).  Let \(E^\star = \{x \in \mathbb{R}^d : (x-c^\star)^\top S^{-1} (x-c^\star) \le 1 \}\) be the minimum volume ellipsoid (of arbitrary center), with condition number at most \(\beta\), that encloses the points \(A\). \vnote{is condition number defined?} \vnote{find and replace for ``aspect ratio" $\rightarrow$ ``condition number"}
    We are guaranteed that there exists a ball \(\widehat{B} \in \mathfrak{B}\), such that \(\widehat{B}\) contains all points in \(E^\star\),
    \[(E^\star \cap A) \subseteq \widehat{B},\]
    and \(\widehat{B}\) has radius at most twice the largest axis length of \(E^\star\).  

    For the sake of analysis, for a candidate ball \(B \in \mathfrak{B}\), rescale the ball and the points in the ball so that \(\mathrm{radius}(B) = 2 \beta\).  For \(\widehat{B}\) this ensures that, in the rescaled space, the maximum axis length of \(E^\star\) is at least \(\beta\), and therefore the minimum axis length of \(E^\star\) is at least \(1\).  That is, \(S^{-1} \preceq I\).  

    For all points \(a_i \in A \cap B\), identify \(a_i \in \mathbb{R}^d\) with the lifted point \(\widetilde{a}_i \in \mathbb{R}^{d + 1}\), 
    \[\wta_i = (a_i~; 1).\]
    Starting with \(\widetilde{A} = \{\wta_i : a_i \in A \cap B\}\) as the included points, the algorithm will do 
    \[J = \lceil c \cdot (\alpha n/\gamma d)\rceil \le 2c \cdot (\alpha n/\gamma d) \] iterations of the following, where \(c\) is the constant from \Cref{lem:outlier-removal-progress}, and the inequality uses (\ref{eq:alpha-n-big}).
    \begin{enumerate}[(a)]
        \item Solve the dual SDP on the included \(\widetilde{a}_i\) to get weights \(0 \le w_i \le 1 \), \(\sum_i w_i = d + 1\) (\Cref{prop:dual-MVE}).
        \item For each included \(i\) independently, discard \(\widetilde{a}_i\) from the included points with probability \(w_i\).  
    \end{enumerate}
    By \Cref{lem:outlier-removal-progress} we have that with probability \(\ge 0.99\), there exists some iteration \(j^\star\) of this procedure on which less than \(\gamma (d+1) \) mass was placed on outliers in the dual solution.  On such an iteration \(j^\star\), \anote{12/18: in the following, is it in $d+1$ dimensions or $d$ dimensions? Also if in $d$ dimensions, should there be a center, or is it actually origin centric?} \vnote{fixed}
    let \(E_{(j^\star)} = \{x \in \mathbb{R}^{d} : (x - c)^\top M_{(j^\star)}^{-1} (x-c) \le 1\} \) be the minimum volume ellipsoid that encloses all of the \(a_i\) that are included on iteration \(j^\star\).
    By \Cref{lem:vol-approx}, we have that 
    \[\log \det (M_{(j^\star)}) \le \log \det (S) + \gamma (d+1) \log (16\beta^2).\]
    
    On iterations subsequent to \(j^\star\), we only remove points \(\widetilde{a}_i\) from being included.  Thus, after the last iteration the dual SDP value is only lower.  
    For an ellipsoid \(E = \{x \in \mathbb{R}^d : (x-c)^\top M^{-1} (x - c) \le 1 \}\), the volume of \(E\) is \(\vol(E) = \kappa_d \det(M)^{1/2}\), where \(\kappa_d\) is the volume of a unit ball in \(d\) dimensions.  Thus
    \[\log \vol (\widehat{E}) \le \log \vol (E^\star) + \gamma (d+1) \log (16\beta^2) .\]

    By \Cref{lem:dont-remove-too-many-points}, we have that with probability \(\ge 0.99\), we remove at most \(c_2 \cdot \frac{\alpha n}{\gamma}\) points \(\wta_i\) over all iterations, for a constant \(c_2\). \(\widehat{B}\) started with at least \((1 - \alpha)n\) points \(a_i\).  
    Thus, with probability \(\ge 0.98\), \(\widehat{E}\) simultaneously achieves approximate coverage
    \[|\widehat{E} \cap A | \ge \left(1 - O\left( \frac{\alpha}{\gamma} \right)\right)n, \]
    and \(\widehat{E}\) approximates \(E^\star\) in volume:
    \[\vol(\widehat{E}) \le \vol(E^\star) \cdot ( 4 \beta )^{2 \gamma (d + 1)}.\]
    
    Finally, when we aggregate the candidate ellipsoids from each of the candidate \(B \in \mathfrak{B}\), conditioned on \(\widehat{E}\) achieving the desired coverage, we will have that the output ellipsoid will achieve coverage and have volume at most that of \(\widehat{E}\).  Thus the output of the algorithm will achieve coverage and volume approximation if this procedure succeeded on at least one of the \(O(\log n)\) trials, which occurs with probability at least 
    \[1 - (0.02)^{O(\log n)} = 1 - o(1).\]

    We analyze the runtime of the algorithm.  There are at most \(n^2\) bounding balls in \(\mathfrak{B}\).  For each ball, checking which points lie within it, and doing the embedding procedure can be done in \(\poly(n, d)\) time.  Then, we solve the dual SDP \(J = O(\frac{\alpha n}{\gamma d})\) times.  Since \(\gamma \ge 1/n\) (see opening paragraph of proof), the number of iterations is \(\poly(n)\).  We solve the relevant SDPs in \(T_{\mathrm{SDP}}\) time.  The outlier removal procedure takes \(O(n)\) time per iteration.  Finally, we solve one last SDP per iteration, and compare \(O(n^2)\) candidate solutions.  Thus, in total, the runtime of the algorithm is \(\poly(n, d) \cdot T_{\mathrm{SDP}}\).
\end{proof}



\section{Algorithm for Subspace Recovery} \label{sec:subspacerecovery}

In this section, we will design an algorithm for the {\em robust subspace recovery} problem using our algorithm for finding the minimum volume ellipsoid with outliers. 
\anote{Can move some of this to the intro: Mention the high-level idea of problem. Relation to PCA. And \cite{HardtMoitra13, BCPV} and also some of the SoS based works, in particular \cite{bakshi2021list}. Also mention (worst-case) hardness in \cite{HardtMoitra13}. }

\paragraph{Robust subspace recovery.} Formally, we are given $n$ points $a_1. \dots, a_n \in \R^d$; without loss of generality we assume they have unit length i.e., $\norm{a_i}_2 =1$ for all $i \in [n]$. We are also given as input 
a closeness parameter $\eps \in (0,1)$, and the outlier fraction $\alpha \in (0,1)$.  The goal is to find a subspace $S \subset \R^d$ (passing through the origin) of smallest possible dimension that is $\varepsilon$-close in Euclidean distance to at least $1-\alpha$ fraction of the given points.\footnote{Given a subspace $S \subset \R^d $ of dimension $s$ given by an (orthogonal) projection matrix $\Pi_S$ of rank $s$, we will say that a point $a \in \R^d$ is $\varepsilon$-close to subspace $S$ iff $\norm{\Pi^{\perp}_S a}_2 = \norm{(I-\Pi_S) a}_2 \le \varepsilon \norm{a}_2$. } 

Unlike previous works, we address the worst-case version of the problem. In particular, we make no genericity, smoothed analysis, or structural assumptions on either the outliers or the inliers, as is done in \cite{HardtMoitra13, BCPV}. More recent works based on sum-of-squares (SoS) relaxations assume distributional assumptions (e.g., Gaussianity) on the points in the subspace~\cite{bakshi2021list}. As mentioned earlier, the worst-case problem is computationally intractable~\cite{HardtMoitra13}; see also Section~\ref{sec:SSE-hard}. Hence we will aim for approximation guarantees i.e., polynomial time algorithms that is (approximately) competitive in terms of all the parameters: the dimension, fraction of outliers $\alpha$, and the closeness parameter $\eps$. 
Note that in the absence of any structural assumptions on the outliers or inliers, our goal is to find a subspace of small dimension that is close to many points; it may not be possible to recover any ground truth subspace.

\begin{figure}[ht]
\begin{tcolorbox}
\begin{center}
    \textbf{Algorithm for Robust Subspace Recovery} 
\end{center}
\textbf{Input:} a set of points (unit vectors) $a_1, \dots, a_n \in \R^d$, parameter $\gamma>0$, 
    closeness parameter $\eps \in (0,1)$\\
\textbf{Output:} a subspace $S \subset \R^d$ passing through the origin
\anote{Need to edit the algorithm etc.}
\begin{enumerate}
\item Set $\epst \coloneq \eps^{4/\gamma}$.  We randomly perturb the points to get $a'_1, \dots, a'_n$ with $\tilde{a}_i = a_i + \zeta_i$ where $\zeta \sim N(0, \frac{\epst^2}{d} I_d)$. Normalize the points to get points $a'_i = \tilde{a}_i / \norm{\tilde{a}_i}_2$ for each $i \in [n]$. 

\item Run Algorithm in Figure~\ref{fig:ellipsoid_tcolor} (from Theorem~\ref{thm:main}) on the points $a'_1, \dots, a'_n$ to find an ellipsoid $\widehat{E}=\{x: x^\top (\widehat{M})^{-1} x \le 1\}$ passing through the origin satisfying the guarantees of Theorem~\ref{thm:main}. \anote{Can directly reference an algorithm for origin-centered version, if there is one.} (Note that to find an origin-centered ellipsoid, we can apply Algorithm in Figure~\ref{fig:ellipsoid_tcolor} to the $2n$ points $\pm a'_1, \pm a'_2, \dots, \pm a'_n$.)   
\item 
Let $\widehat{S}$ be the top singular space of $\hat{M}$ spanned by all the singular vectors with corresponding singular value $\ge \eps^2$. Return $\widehat{S}$. 
\end{enumerate}
\end{tcolorbox}
\caption{Algorithm for finding a subspace of dimension at most $(1-\gamma)d$ that contains at least $1-\alpha'$ fraction of points}
\label{fig:subspacerecovery}
\end{figure}

The above algorithm runs our robust minimum volume enclosing ellipsoid algorithm on a small random perturbation applied to our given points. The random perturbation is performed by our algorithm, and it is applied to all the points.  We remark that this is different from the smoothed analysis setting where the algorithm is given randomly perturbed points as input, where the points in the subspace are assumed to be randomly perturbed in a way that keeps them inside the (unknown) ground-truth subspace. In contrast, our algorithm applies to an arbitrary set of points given as input.

\begin{theorem}\label{thm:subspacerecovery}
Let $\alpha \in (0,1)$, and $\gamma \in (0,1)$, $\epsilon \in (0,1)$ be input parameters. 
There exists universal constants $c_1, c_2>0, C>1$ such that the following holds for  $\epst \leq \eps^{4/\gamma}$. Suppose we are given as input $n \ge C d \log (d/\epst) + \gamma d / \alpha$  unit vectors $A=\{a_1, \dots, a_n\} \in \R^d$, with a subspace $S^\star$ of dimension $(1-8\gamma) d$ that is $\epst$ close to at least $(1-\alpha)n$ of the given points. Then the algorithm in Figure~\ref{fig:subspacerecovery} runs in polynomial time in $n,d$ and finds w.h.p. a subspace $S \subset \R^d$ of dimension at most $(1-\gamma)d$ such that at least $(1-\alpha')n$ points are $\eps$ close to $S$, for $\alpha'= c_2 \alpha/ \gamma$. 
\end{theorem}
\anote{12/18: Added} We remark that the lower bound condition on the number of samples $n$ is essentially without loss of generality since we can just resample (or take multiple copies of the dataset) and run the same algorithm to get a similar guarantee; this dependence will feature in the running time.  
When we know that there are many points that lie exactly on the subspace $S$, we have $\epst=0$. So we can pick $\eps$ to be whatever desired accuracy (that can be arbitrarily close to $0$) and since $\epst= 2^{-\log(1/\eps)/\gamma}$ (note that the number of bits of precision is only $\log(1/\eps)/\gamma$), and apply the above algorithm to find a subspace that is $\eps$-close to the desired number of points. 
\anote{Ensure that the main algorithm uses ellipsoid or related techniques to get $\poly(n,L)$ time with $L$ bits of precision.}

Before we proceed to the proof, we first need the following claim that uses the random perturbations to the points to prove that any ellipsoid that encloses a significant fraction of the points needs to sufficiently fat.

\begin{lemma}\label{lem:fatness:random}
In the notation of Algorithm~\ref{fig:subspacerecovery}, there exists universal constants $C,c,c' >0$ such that when $n \ge Cd \log(d/\epst)$ we have with probability at least $1-\exp(-c' n) - n\exp(-2d)$ (over the random perturbation of the points) that any origin-centered ellipsoid $E=\{x: x^\top M^{-1} x \le 1\}$ that encloses at least $4n/5$ of the points $a'_1, \dots, a'_n$ has $\lambda_{\min}(M) \ge c (\epst)^2/d$. 
\end{lemma}

This lemma follows directly from the following anti-concentration claim that shows that in every direction most points are far from the origin. 
\begin{claim}\label{claim:fatness:helper}
In the notation of Lemma~\ref{lem:fatness:random} for some universal constant $c,c'>0$, we have with probability $1-\exp(-c'n) - n\exp(-2d)$, 
\begin{equation} \label{eq:fatness:helper}
\forall u \in \R^d \text{ s.t } \norm{u}=1, ~ \Big|\Big\{ i \in [n]: \iprod{a'_i, u}^2 \le \frac{ \epst^2}{256d}\Big\} \Big| \le 0.2 n. 
\end{equation}
\end{claim}
\noindent We defer the proof of Claim~\ref{claim:fatness:helper} to later. The choice of the constant $0.2$ is arbitrary, and can be chosen as desired.  
\begin{proof}[Proof of Lemma~\ref{lem:fatness:random}]
    We condition on the event \eqref{eq:fatness:helper} in Claim~\ref{claim:fatness:helper}. 
    Consider any ellipsoid $E=\{x: x^\top M^{-1} x \le 1\}$ containing $T \subseteq \{a'_1, \dots, a'_n\}$ with $|T|\ge 4n/5$. Let $\lambda$ be the smallest eigenvalue of $M$, and let $u$ be a corresponding eigenvector of unit length. Then 
    $$ \forall i \in T, ~ 1 \ge (a'_i)^\top M^{-1} a'_i \ge \frac{\iprod{a'_i, u}^2}{\lambda}~~ \implies \forall i \in T, ~ \iprod{a'_i, u}^2 \le \lambda. $$   
    From Claim~\ref{claim:fatness:helper}, we conclude that $\lambda \ge \frac{\epst^2}{256d}$.  
\end{proof}

We now proceed to the proof of the main theorem of the section. 
\begin{proof}[Proof of Theorem~\ref{thm:subspacerecovery}]
Set $\gammast=8\gamma$, $\epst=\eps^{4/\gamma}$. 
Suppose there exists a subspace $S^\star$ (passing through the origin) of dimension $(1-\gammast)d$ that is $\epst$-close to $(1-\alpha)n$ points in $A$.  Let $\Pi \in \R^{d\times d}$ be the projection matrix onto the subspace $S^\star$ and $\Pi^{\perp} = I-\Pi$ be the projection onto the subspace orthogonal to $S^\star$.  

\noindent We first show that there is a well-conditioned ellipsoid $E^\star$ of small volume containing the points that are $\epst$-close to $S^\star$. Consider the ellipsoid $E^\star$ given by
\begin{equation}\label{eq:rsubspace:ellip}
E^\star=\Big\{ x \in \R^d: \frac{x^{\top} \Pi x}{2}+ \frac{x^{\top} \Pi^{\perp} x}{2\epst^2} \le 1 \Big\}.
\end{equation}
We have $\norm{a_i}_2=1$ for all points $i \in [n]$. If $a_i$ is $\epst$-close to $S^\star$, then 
\begin{align*}
\norm{\Pi^{\perp} a_i}_2  \le \epst~~
&\implies ~ \frac{a_i^\top \Pi a_i}{2} + \frac{a_i^\top \Pi^{\perp} a_i}{2\epst^2} \le 1. \\
&\implies ~ a_i \in E^{\star}. 
\end{align*}
Moreover,  the ellipsoid $E^\star$ has condition number $\beta=1/\epst$, and its volume is 
$$\vol(E^\star) = \kappa_d \cdot 2^{(1-\gammast)d/2}(2\epst^2)^{\gammast d/2}  = \kappa_d \cdot  2^{d/2} \epst^{\gammast d},$$ 
where $\kappa_d$ is the volume of a unit radius ball in $d$ dimensions. 

By running the algorithm from Theorem~\ref{thm:main} on these points, we get an origin-centered ellipsoid \anote{double check it is justified/referenced correctly} $\widehat{E} = \{x: x^\top \widehat{M}^{-1} x \le 1\}$. Let $T=\{i \in [n]: a'_i \in \widehat{E}\}$; then $|T| \ge (1-c' \alpha/\gamma)n$, and  
$$\vol(\widehat{E}) \le \beta^{\gamma d} \vol(E^\star) \le \kappa_d \cdot 2^{d/2} \epst^{(\gammast  -\gamma)d } \le \kappa_d \cdot 2^{d/2} \epst^{\frac{7}{8}\gammast d }, $$
since $\gammast = 8\gamma$. 

Let $\widehat{S}$ be the subspace spanned by the eigenvalues of $\widehat{M}$ larger than $\eps$, and let $\widehat{\Pi}$ be the projection matrix onto the subspace $\widehat{S}$, and $\widehat{\Pi}^\perp$ be the projection onto the subspace orthogonal to $\widehat{S}$. Then 
\begin{align*}
    \forall i \in T,~&~ 1 \ge \norm{\widehat{M}^{-1/2} a'_i}_2^2 \ge \frac{\norm{\widehat{\Pi}^{\perp} a'_i}_2^2}{\eps^2} ~\implies~ 
    \norm{\widehat{\Pi}^{\perp} a'_i}_2^2 \le \eps^2.
\end{align*}
Hence every point $a'_i$ for each $i \in T$ is $\eps$-close to $\widehat{S}$, as required. By triangle inequality, along with standard length concentration of Gaussians, we have that $a_i$ is $\eps+ O(\epst) \le 2 \eps$-close to $\widehat{S}$.

It remains to lower bound the co-dimension of $\widehat{S}$. Let the codimension $\hat{k} \coloneq d - \dim(\widehat{S})$.  Let $\widehat{M}$ have eigenvalues $\hat{\lambda}_1 \ge \hat{\lambda}_2 \ge  \dots \ge  \hat{\lambda}_d \ge 0$; then  $\hat{\lambda}_{d-\hat{k}} \ge \eps^2$. From Lemma~\ref{lem:fatness:random}, we have $\hat{\lambda}_d \ge c_2 \epst^2/d$ w.h.p.. Conditioning on this, we have  
\begin{align}
\kappa_d\cdot 2^{d/2} \epst^{\frac{7}{8}\gammast \cdot d}&\ge \vol(\widehat{E}) = \kappa_d \prod_{i =1}^d \hat{\lambda}_i \nonumber \\
2^{d/2} \epst^{\frac{7}{8}\gammast \cdot d}& \ge  \prod_{i \le d-\hat{k}} \hat{\lambda}_i \cdot \prod_{i: i \ge d-\hat{k}+1} \hat{\lambda}_i \ge (\eps^2)^{d-\hat{k}} (\hat{\lambda}_d)^{\hat{k}}\nonumber \\
&\ge \eps^{2d} \cdot \Big(\frac{c_2 \epst^2}{d}\Big)^{\hat{k}} \ge \epst^{\gamma d} \cdot \Big(\frac{c_2 \epst^2}{d}\Big)^{\hat{k}}, \label{eq:rsr:help1}
\end{align}
where the last inequality is from $\epst \leq \eps^{4/\gamma}$. 

Since $\epst \leq \eps^{4/\gamma}$ and $\gammast = 8\gamma$, we have $\epst^{\gammast/4} = \epst^{2\gamma} \leq \eps^{8}$. Thus, by picking closeness parameter $\eps \leq 2^{-1/8}$, we have 
$$2^{d/2} \epst^{\frac{7}{8}\gammast \cdot d} \leq \left(2 \epst^{\gammast/4}\right)^{d/2} \epst^{\frac{3}{4} \gammast d} \leq  (2\eps^{8})^{d/2} \epst^{\frac{3}{4} \gammast d}
\leq \epst^{\frac{3}{4} \gammast d}.$$
Putting this together with \eqref{eq:rsr:help1}, we get 
\begin{align*}
\epst^{\frac{3}{4} \gammast d} &\ge \epst^{\gamma d} \cdot \Big(\frac{c_2 \epst^2}{d}\Big)^{\hat{k}} \ge \epst^{\frac{1}{8}\gammast d} \cdot \Big(\frac{c_2 \epst^2}{d}\Big)^{\hat{k}}\\
\Big(\frac{c_2 \epst^2}{d}\Big)^{\hat{k}} & \le \epst^{\frac{5}{8} \gammast d} ~~ \implies ~ \hat{k} \ge \frac{5}{8} \gammast d  \cdot \Big(\frac{\log(1/\epst)}{2\log(1/\epst)+ \log(d/c_2)}\Big) \geq \gamma d,
\end{align*}
where the last inequality is due to $\log(1/\epst) \geq \log(1/\eps^{4/\gamma}) \geq \log(d/c_2)$ by picking $\eps$ sufficiently small.
This lower bounds the codimension of $\widehat{S}$. This completes the proof. 
\end{proof}

We now complete the proof of the helper claim. 
\begin{proof}[Proof of Claim~\ref{claim:fatness:helper}]
This follows from a relatively standard net-based argument. Set $\eta \coloneqq \frac{\epst}{32\sqrt{d}}$, and consider an $\eta$-net $\calN_\eta$ of unit vectors in $\R^d$; $|\calN_\eta| \le \big(\frac{2+\eta}{\eta} \big)^d$. For a fixed unit vector $u \in \calN_\eta$ and a fixed $i \in [n]$, the small ball probability (or anticoncentration) of the Gaussian $\zeta \sim N(0,\frac{\eps_1^2}{d} I))$, 
\begin{align*} 
\Pr\Big[ |\iprod{a'_i, u}| \le \frac{\eps_1}{8\sqrt{d}} \Big] &= \Pr_{g \sim N(0,1)}\Big[ \big|\iprod{a_i, u}+ \frac{\eps_1}{\sqrt{d}} g\big| \le \frac{\eps_1}{8\sqrt{d}} \Big]   \nonumber\\
&\le \sup_{t \in \R} \Pr_{g \sim N(0,1)}\Big[| g+t| \le \frac{1}{8} \Big] \le 0.1 \nonumber. \\  
\text{For a fixed } u \in \calN_\eta,& ~ \Pr\Big[ \Big|\Big\{i \in [n]: |\iprod{a'_i, u}| \le \frac{\eps_1}{8\sqrt{d}} \Big\} \Big| \ge 0.2 n\big| \Big] \le \exp(-c_1 n).
\end{align*}

 By doing a union bound over all $u \in \calN_\eta$, we have that since $n \ge \frac{2}{c_1} \log(\calN_\eta)$, we have  
\begin{align}
\Pr\Big[ \sup_{u \in \calN_\eta}\Big|\Big\{i \in [n]: |\iprod{a'_i, u}| \le \frac{\eps_1}{8\sqrt{d}} \Big\} \Big| \ge 0.2 n\big| \Big] \le |\calN_\eta| \cdot \exp(-c_1 n) \le \exp(-\tfrac{1}{2}c_1 n).
\label{eq:unionboundsmooth}
\end{align}

Let us condition on the above event in \eqref{eq:unionboundsmooth}. For any unit vector $v \in \R^d$, there exists $u^* \in \calN_\eta$ with $\norm{v-u^*}_2 \le \eta$. Moreover $\norm{a'_i} \le \norm{a_i}+ \norm{\zeta_i}$ where $\zeta_i \sim N(0,\frac{\epst^2}{d}I)$; so with probability at least $1-n \exp(-2d)$, we have $\forall i \in [n], ~\norm{a'_i} \le 1+ 2\eps_1 \le 2$. Hence $\forall i \in [n], ~|\iprod{a'_i, v} - \iprod{a'_i, u^*}|\le 2 \eta \le \frac{\eps_1}{16\sqrt{d}}$. Thus we conclude that with high probability
\begin{align}
\Pr\Big[ \sup_{v \in \R^d: \norm{v} =1}\Big|\Big\{i \in [n]: |\iprod{a'_i, v}| \le \frac{\eps_1}{16\sqrt{d}} \Big\} \Big| \ge 0.2 n\big| \Big] \le  \exp(-\tfrac{1}{2}c_1 n).
\nonumber
\end{align}


\end{proof}


\section{Computational Intractability}\label{sec:SSE-hard}
In this section, we provide evidence of the computational intractability of the minimum volume confidence ellipsoid problem. Specifically, under the Small Set Expansion (SSE) hypothesis of~\cite{raghavendra2010graph}, we show that this problem is hard to approximate. The SSE hypothesis is closely related to the Unique Games Conjecture~\citep{khot2003UGC}.

\begin{conjecture}[SSE hypothesis of~\cite{raghavendra2010graph}] \label{conj:sse} For any $\varepsilon>0$, there is a constant $\delta\in (0,1)$ such that there is no polynomial time algorithm to distinguish between the following two cases given a graph $G=(V,E)$ on $n$ vertices with degree $\Delta$:
\begin{itemize}
\item YES: Some subset $S\subseteq V$ with $|S| = \delta n$ satisfies that the induced subgraph on $S$ is dense i.e., the number of edges going out of $S$ is $|E(S,V\setminus S)| \le \varepsilon \Delta |S|$ edges.  
\item NO: Any set $S\subseteq V$ with $ \tfrac{1}{2}\varepsilon \delta n \le |S| \le 2\delta n$ has most of the edges incident on it going outside i.e., $|E(S,V\setminus S)| \ge (1-\varepsilon)  \Delta|S|$.
\end{itemize}
\end{conjecture}
Define the edge expansion of a $\Delta$-regular graph $G$ at scale $\delta$ by 
$$
\phi_G(\delta)\ := \min_{S: |S| \leq \delta n} \frac{|E(S, V\setminus S)|}{\Delta|S|}.
$$
Then, the SSE hypothesis asserts that it is computationally hard to distinguish between: (1) YES case, where $\phi_G(\delta) \leq \varepsilon$, and (2) NO case, where $\phi_G(\delta') \geq 1- \varepsilon$ for all $\delta'\in[\delta\varepsilon/2,2\delta]$.

We now show that the minimum volume confidence ellipsoid problem is SSE-hard to approximate. 
In particular, for ellipsoids covering at least a \((1-\alpha)\) fraction of the points, it is SSE-hard to obtain a volume approximation factor better than \(\beta^{\gamma d/2}\) for any fixed \(\gamma<\delta\). 
This hardness remains true even if one allows a small slack in the coverage requirement, replacing \(1-\alpha\) by \(1-(1+\delta/2)\alpha\).

\begin{theorem}[SSE--hardness]\label{thm:SSE-hard}
Fix any $\varepsilon, \delta \in (0,0.1)$ and $\gamma \in (0,\delta)$ and $\alpha \in (0,1-(1-\varepsilon)\delta)$.
For any $C \ge 5/\delta$, given a graph $G$, we can construct in
polynomial time a set $A=\{a_1,\ldots,a_n\} \subseteq \R^d$ with the property:

\begin{itemize}
    \item 
YES: If $\phi_G(\delta)\leq \varepsilon$,
then there exists an $\beta$-conditioned ellipsoid covering $(1-\alpha) n$ points in $A$ with volume at most $K$; and 
    \item 
NO: If
$\phi_G(\delta')\ge 1-\varepsilon$ for all $\delta'\in[\delta\varepsilon/2,2\delta]$,
then every ellipsoid covering $(1-(1+\delta/2)\alpha)n$ points has volume at least ${\beta}^{\gamma d/2}\,K$, 
\end{itemize}
where $\beta = 1/\widetilde\varepsilon$,
$
K\ = \ \kappa_d\,
2^{d/2}\,\widetilde\varepsilon^{(1-\delta) d},\ \widetilde\varepsilon\ =\ d^{-C},
$
and $\kappa_d$ is the volume of the Euclidean unit ball in $\R^d$.

Consequently, unless the Small-Set Expansion conjecture fails, no
polynomial-time algorithm can approximate the volume of the optimal $\beta$-conditioned ellipsoid with coverage $1-\alpha$
within a factor of $\beta^{\gamma d/2}$ for any $\gamma \in (0,\delta)$.
\end{theorem}

To prove the above SSE-hardness result, we first use the reduction from the Small Set Expansion to the robust subspace recovery in Theorem 14 by~\cite{HardtMoitra13}. Their construction is randomized and yields hardness in the regime where one seeks a subspace containing only a $(1-\varepsilon)\delta$ fraction of points. We make two minor modifications. First, we make the construction deterministic. Second, by padding the instance with additional points at the origin, we extend the reduction to the higher-coverage regime $1-\alpha$ for any $\alpha \in (0,1-(1-\varepsilon)\delta)$. The resulting statement is as follows.

\begin{theorem}\label{thm:HM-modified}
Let $\varepsilon, \delta \in (0,0.1)$ and $\alpha \in (0,1-(1-\varepsilon)\delta)$. There is a polynomial-time reduction mapping a
$\Delta$-regular graph $G$ to an instance $U=\{u_1,\ldots,u_n\}\subset\R^{d}$ of unit vectors such that:
\begin{itemize}
\item YES: If $\phi(\delta)\le\varepsilon$, there exists a subspace
$S^\star$ of dimension $\delta d$ containing at least $1-\alpha$ fraction of the points.
\item NO: If $\phi(\delta')\ge 1-\varepsilon$ for every
$\delta'\in[\delta\varepsilon/2,\,2\delta]$, then for every $\delta' \in[\delta\varepsilon/2,\,2\delta]$,  every subspace of dimension
$\delta' d$ contains at most $1-(1+(1-9\varepsilon)\delta)\alpha$ fraction of the points.
\end{itemize}
\end{theorem}

\begin{proof}[Proof of Theorem~\ref{thm:HM-modified}]
We construct the instance $U$ similar to the reduction by~\cite{HardtMoitra13} as follows. Let $G(V,E)$ be a $\Delta$-regular graph on $d$ vertices.
For each edge $e = (i,j) \in E$, let $u'_e = (e_i + e_j)/ \sqrt{2}$ be the unit length vector corresponding to $e$. 
Let $U'=\{u'_1,\ldots,u'_m\}\subset\mathbb{R}^{d}$ be the instance constructed from $G$
(with $d=|V|$ and $m=|E|=\Delta d/2$). Fix any parameter $\eta \ge 0$ which is specified later and set $T = \eta m / \delta$. Let $U$ be the set of points in $U'$ and $T$ points located at the origin. Then, $U$ contains $|U| = m + T$ points. 

\paragraph{Completeness.}
Assume $\phi_G(\delta)\le\varepsilon$. Then, by Theorem 14 in \cite{HardtMoitra13}, there exists a
$\delta d$-dimensional subspace $S^\star$ containing at least $(1-\varepsilon)\delta m$ points of $U'$.
Moreover, $S^\star$ is a coordinate subspace spanned by a subset of the standard basis vectors.
Thus, all $T$ appended points at the origin lie in $S^\star$ as well. Therefore, we have for our new instance $U$ 
\[
\frac{|U \cap S^\star|}{|U|}
\ \ge\ \frac{T+(1-\varepsilon)\delta m}{m+T}
\ =\ 1-\frac{m-(1-\varepsilon)\delta m}{m+T}
\ \ge\ 1-\frac{\delta-(1-\varepsilon)\delta^2}{\delta+\eta},
\]
where the last inequality is from $T = \eta m/\delta$. For any $\alpha \in (0,1 - (1-\varepsilon)\delta)$, we choose the number of appended points $T = \eta m /\delta$ such that $\eta$  satisfies $\alpha = \frac{\delta-(1-\varepsilon)\delta^2}{\delta+\eta}$. Then, this subspace $S^\star$ contains at least $1-\alpha$ fraction of the points in $U$.

\paragraph{Soundness.}
Assume $\phi_G(\delta')\ge 1-\varepsilon$ for all $\delta'\in[\delta\varepsilon/2,\,2\delta]$.
Fix any $\delta'\in[\delta\varepsilon/2,\,2\delta]$ and any subspace $S$ of dimension $r = \delta'd$. We use the following claim to bound the number of points in $U$ covered by this subspace $S$.

\begin{claim}\label{clm: dim-slack}
For every integer $r$ with  $\frac{\delta\varepsilon}{2}\,d \le r \le 2\delta d$, every set
$P\subseteq U'$ that lies in some subspace of dimension at most $r$ satisfies
\[
\frac{|P|}{m}\ \le\ 4\varepsilon\,\frac{r}{d}.
\]
Equivalently, any $r$-dimensional subspace contains at most a $4\varepsilon\cdot(r/d)$ fraction of the points in $U'$.
\end{claim}

By Claim~\ref{clm: dim-slack},
\[
|U' \cap S|\ \le\ 4\varepsilon\,\frac{r}{d}\,m\ =\ 4\varepsilon\delta'\,m\ \le\ 8\varepsilon\delta\,m.
\]
Hence
\[
\frac{|U \cap S|}{|U|}
\ \le\ \frac{T+8\varepsilon\delta m}{m+T}
\ =\ 1-\frac{m-8\varepsilon\delta m}{m+T}
\ \le\ 1-\frac{\delta-8\varepsilon\delta^2}{\delta+\eta}.
\]
Since $\varepsilon \in (0,0.1)$, we have $1-\varepsilon \geq 8\varepsilon$. Then, we have 
$$
\frac{|U \cap S|}{|U|} \leq 1-\frac{\delta-8\varepsilon\delta^2}{\delta+\eta} = 1-\alpha -\frac{(1-9\varepsilon)\delta^2}{\delta + \eta} \leq 1-\alpha - (1-9\varepsilon)\delta \alpha,
$$
where the last inequality is from $\alpha \leq \delta/(\delta+\eta)$.
Thus, the subspace $S$ contains at most $1- (1+(1-9\varepsilon)\delta)\alpha$ fraction of points in $U$.
\end{proof}

\begin{proof}[Proof of Claim~\ref{clm: dim-slack}]
Let $P\subseteq U$ be contained in a subspace of dimension at most $r$, let $F\subseteq E$
be the corresponding edge set, and let $S:=N_B(F)\subseteq V$ be the neighborhood of $F$ in the
bipartite incidence graph $B=(E,V)$. 

We require the following claim, corresponding to Claim 3.4 in \cite{HardtMoitra13}. As their argument relies on a randomized construction, we include a proof using our deterministic construction in Appendix~\ref{apx:SSE}.

\begin{claim}\label{clm:span}
    For every set of points $P \subseteq U$ corresponding to a set of edges $F \subseteq E$, we have
    $$
    \frac{|N_B(F)|}{2}\ \leq\ \dim(\mathrm{span}\,P)\ \le\ |N_B(F)|.
    $$
\end{claim}

By Claim~\ref{clm:span}, we have
\begin{equation}\label{eq:span-vs-neighborhood}
\frac{|S|}{2}\ \leq\ \dim(\mathrm{span}\,P)\ \le\ |S|.
\end{equation}
Then, we have $r\leq |S| \leq 2 r$.

We consider two cases.

\noindent\emph{Case 1: $|S|\le 2\delta d$.}
Since $|S| \geq r \geq \delta \varepsilon d$, let $\delta'=|S|/d\in[\delta\varepsilon/2,2\delta]$.
By the small-set expansion hypothesis,
the number of edges leaving $S$ is at least $|E(S,V\setminus S)| \geq (1-\varepsilon)\Delta|S|$.
Since $\Delta|S|=2|E(S,S)|+|E(S,V\!\setminus\!S)|$, we obtain
\[
|E(S,S)|\ \le\ \frac{\varepsilon\Delta}{2}\,|S|.
\]
As $F\subseteq E(S,S)$ and $|S|\le 2r$, we get
\[
|P|=|F|\ \le\ \frac{\varepsilon\Delta}{2}\,|S|\ \le\ \varepsilon\Delta r.
\]
Dividing by $m=\Delta d/2$ yields $|P|/m\le 2\varepsilon\,(r/d)$.

\noindent\emph{Case 2: $|S|>2\delta d$.}
Set $s:=2\delta d$ and let $S'\subseteq S$ be a size-$s$ subset maximizing $|E(S',S')|$.
By averaging (equivalently, taking expectation over a uniform size-$s$ subset of $S$),
\begin{equation}\label{eq:subset-dense}
|E(S',S')|\ \ge\ \frac{\binom{s}{2}}{\binom{|S|}{2}}\ |E(S,S)|
\ \ge\ \Big(\frac{s}{|S|}\Big)^{\!2}|E(S,S)|
\ \ge\ \Big(\frac{s}{|S|}\Big)^{\!2}|F|.
\end{equation}
Applying the expansion hypothesis to $S'$ (note $|S'|=2\delta d$) gives
\[
|E(S',S')|\ \le\ \frac{\varepsilon\Delta}{2}\,|S'|
\ =\ \varepsilon\Delta\,\delta d.
\]
Combining with \eqref{eq:subset-dense} and using $|S|\le 2r$ from \eqref{eq:span-vs-neighborhood},
\[
|P|=|F|\ \le\ \Big(\frac{|S|}{s}\Big)^{\!2}|E(S',S')|
\ \le\ \Big(\frac{2r}{2\delta n}\Big)^{\!2}\cdot \varepsilon\Delta\,\delta d
\ =\ \varepsilon\Delta\,\frac{r^2}{\delta n}.
\]
Hence
\[
\frac{|P|}{m}\ \le\ \frac{\varepsilon\Delta\,r^2/(\delta d)}{\Delta d/2}
\ =\ 2\varepsilon\,\frac{r^2}{\delta d^2}
\ \le\ 4\varepsilon\,\frac{r}{d},
\]
where the last step uses $r/d\le 2\delta$ since $r\le 2\delta d$.

Combining both cases yields the claimed bound $|P|/m\le 4\varepsilon\,(r/d)$.
\end{proof}

\begin{proof} [Proof of Theorem~\ref{thm:SSE-hard}]
We construct a set of points $A$ as follows.
Given a Small Set Expansion instance $G(V,E)$, let $U=\{u_i\}_{i=1}^n$ be the set of points constructed by Theorem~\ref{thm:HM-modified} with
parameters $(\varepsilon,\delta)$. 
We then add random perturbations to these points and normalize them to unit vectors. 
Independently sample $z_i\sim
\mathcal{N}(0,\widetilde\varepsilon^2 I_d/d)$ with $\widetilde\varepsilon=d^{-C}$ for
a fixed constant $C\ge 10$, and set $a_i=(u_i+z_i)/\|u_i+z_i\|$. We then consider the set of points $A=\{a_i\}$ and target coverage
$1-\alpha$ for $\alpha \in (0,1-(1-\varepsilon)\delta)$.

\paragraph{Completeness.}
If $\phi_G(\delta)\le\varepsilon$, then by Theorem~\ref{thm:HM-modified}, there exists a subspace
$S^\star\subset\R^{d}$ of dimension $\delta d$ containing at least $(1-\alpha)n$ points of $U$.

Let $\Pi$ be the projection to this subspace $S^*$ of dimension $\delta d$. Then, for any
$\widetilde\varepsilon\in(0,1)$ define the ellipsoid 
\begin{equation}\label{eq:star-ellipsoid}
E^\star \ :=\ \Big\{x\in\R^{d}:\ \frac{x^\top \Pi x}{2}
+\frac{x^\top \Pi^\perp x}{2(\widetilde\varepsilon)^2}\le 1\Big\}.
\end{equation}
If a unit vector $a$ is $\widetilde\varepsilon$-close to $S$ then $a\in E^\star$; and
\begin{equation}\label{eq:vol-star}
\vol\!\left(E^\star \right)=\kappa_d\,
2^{\delta d/2}\,(2\widetilde\varepsilon^2)^{(1-\delta) d/2} = \kappa_d\,
2^{d/2}\,\widetilde\varepsilon^{(1-\delta) d},
\end{equation}
where $\kappa_d$ is the volume of the Euclidean unit ball in $\R^{d}$.

For any points $u_i \in  S$, and by a Gaussian tail,
$\|z_i\|_2\le \widetilde\varepsilon$ w.h.p., hence $\dist(a_i, S)\le \widetilde\varepsilon$.
Thus, we have all such $a_i$ lie in the $\beta$-conditioned ellipsoid $E^\star$ with volume $K$.

\paragraph{Soundness.}
We now consider the NO case. Suppose, toward contradiction, that there exists
an ellipsoid $E=\{x:x^\top M^{-1}x\le 1\}$ such that for $\alpha' = (1+\delta/2)\alpha \geq (1+(1-9\varepsilon)\delta)\alpha - 40\delta d/ n$
\begin{equation}\label{eq:assumed-small}
|E\cap A|\ \ge\ (1-\alpha') n \qquad\text{and}\qquad
\vol(E)\ \le\ \beta^{\gamma d/2}\,K.
\end{equation}

Let $E^*$ be a minimum-volume ellipsoid in $\R^{d}$ covering at least $(1-\alpha')$ fraction of the points in $A$, and let $M^*$ be its shape matrix. Then
\[
\vol(E^*)\ \le\ \vol(E)
\quad\Longrightarrow\quad
\det(M^*)\ 
\ \le\ \Big(\frac{\vol(E)}{\kappa_d}\Big)^2.
\]
Since $\vol(E) \leq \beta^{\gamma d/2} K$, we have
\begin{equation}\label{eq:det-upper-D}
\det(M^*)\ \le\ \beta^{\gamma d}\,2^{d}\,
\widetilde\varepsilon^{2(1-\delta) d}.
\end{equation}

We now bound the small eigenvalue directions in $M^*$.
Let the eigenvalues of $M^*$ be $\lambda_1\ge\cdots\ge \lambda_d>0$.
Fix a constant threshold $\tau:=\varepsilon_0^2$ with $\varepsilon_0\in(0,0.01)$, and let
$\widehat{T}\subset\R^{d}$ be the span of eigenvectors with eigenvalues $\lambda_i\ge \tau$. Let
$\widehat{k}_d:=d-\dim(\widehat{T})$ be the number of small eigenvalues and $\Pi^\perp_{\widehat{T}}$ be the projection to orthogonal directions of $\widehat{T}$.
For any point $a_i$ contained in $E^*$, we have $1\ge a_i^\top (M^*)^{-1} a_i\ge \|\Pi^\perp_{\widehat{T}}a_i\|_2^2/\tau$. Thus, 
every covered point is $\varepsilon_0$-close to the subspace $\widehat{T}$; hence at least $(1-\alpha')n$ of
the points in $A$ satisfy $\dist(\tilde{a}_i,\widehat{T})\le \varepsilon_0$.
By Lemma~\ref{lem:fatness:random}, with high probability
$\lambda_{\min}(M^*)\ge c\widetilde\varepsilon^2/d$ for some constant $c$, so
\begin{equation}\label{eq:prod-bounds-D-unit}
\tau^{\,d-\widehat{k}_d}\Big(\tfrac{c\widetilde \varepsilon^2}{d}\Big)^{\widehat{k}_d}
\ \le\ \det(M^*)
\ \le\ \beta^{\gamma d}\,2^{d}\,
\widetilde\varepsilon^{2(1-\delta) d}.
\end{equation}
Let $Y:=\log d-(\log (c \widetilde \varepsilon^2) -\log\tau) >0$. Taking logs in the first inequality of Equation~\eqref{eq:prod-bounds-D-unit} and rearranging gives
\begin{equation}\label{eq:kD-lb}
\widehat{k}_d\ \ge\ \frac{-\log\det(M^*)+ d \log\tau }{Y}.
\end{equation}
Taking logs in the second inequality of Equation~\eqref{eq:prod-bounds-D-unit} and $\widetilde \varepsilon = d^{-C}$, we have 
\begin{align*}
\log\det(M^*)&\leq \log( 2^{d} \widetilde \varepsilon^{(2-2\delta-\gamma) d})  \leq  d\,\log 2 + \left(2-2\delta-\gamma\right)d \log\widetilde\varepsilon\\
&
= -\left(2-2\delta-\gamma\right)d\ C\log d + O(d).
\end{align*}
Thus, we have $-\log\det(M^*)\ge \left(2-2\delta-\gamma\right) d \cdot C\log d - O(d)$.
We have $Y=\log d - (\log (c \widetilde \varepsilon^2) -\log \tau)=(2C+1)\log d + O(1)$. 
Therefore, by Equation~\eqref{eq:kD-lb}, we have
for sufficiently large $d$,
\begin{equation}\label{eq:kD-final}
\widehat{k}_d\ \ge\ \frac{\left(2-2\delta-\gamma\right)d \cdot C\log d - O(d)}{(2C+1)\log d + O(1)}
\ \ge\ \left(1-\delta-\frac{ \gamma }{2}\right)\frac{C}{C+1}d - o(d).
\end{equation}
Then, we have
\begin{equation}\label{eq:That-dim}
\dim(\widehat{T})\ \leq\ d -\widehat{k}_d \leq \left(\delta+\frac{ \gamma }{2} + \frac{1}{C}\right)d + o(d).
\end{equation}

Since $\gamma \in (0,\delta)$, for sufficiently large $C > 5/\delta$, we have $\dim(\widehat{T}) \leq 7/4\cdot \delta d$. Thus, we have at least $(1-\alpha')$ fraction of the points in $A$ are $\varepsilon_0$ close to this subspace $\widehat{T}$ with dimension less than $7/4\cdot \delta d$.

We now show that for these points $a_i$ that are close to this subspace $\widehat{T}$, a large fraction of the points $u_i$ before random perturbation lie exactly in a subspace with dimension at most $2\delta d$.
Since $a_i = (u_i + z_i)/\|u_i +z_i\|$, we have 
$$
\|a_i - u_i\|\leq \|u_i-(u_i + z_i)\| + \|(u_i + z_i) -a_i\| \leq \|z_i\| + |1-\|u_i+z_i\||\leq 2\|z_i\|.
$$
Thus, we have these $u_i$ are $\varepsilon_0 + 2\widetilde \varepsilon$ close to the subspace $\widehat{T}$.
In Lemma~\ref{lem:exact-contain} below, we show that when many points are close to a subspace in this instance, then there exists a subspace that exactly contains a large portion of these points. 
By Lemma~\ref{lem:exact-contain} and choosing $\theta^2 = 7/8$, we get a subspace with dimension at most $2\delta d$ that contains at least $(1-\alpha')n - 40\delta d$ points in $U$. Let $\alpha'' = \alpha' +40\delta d/ n = (1+(1-9\varepsilon)\delta)\alpha$. 

Thus, we have there is a subspace $\overline{S}\subset\R^d$ with
$\dim(\overline{S})\le 2\delta d$ that contains at least $(1-\alpha'')$ of the points in $A$. This contradicts the NO case in Theorem~\ref{thm:HM-modified}.
Therefore, every ellipsoid covering $(1-\alpha') n$ points has
volume at least $\beta^{\gamma d}K$.
\end{proof}

\begin{lemma}[Exact containment from many near points]\label{lem:exact-contain}
Let $S\subset\R^d$ be an $r$-dimensional subspace and set $P:=P_S$ to be the projection matrix onto the subspace $S$. Fix $\varepsilon\in(0,1/4)$ and define the set of $\varepsilon$-close edges
\[
F\ :=\ \bigl\{\{i,j\}\in E:\ \|(I-P)u_{ij}\|_2\ \le\ \varepsilon\bigr\},
\]
where $u_{ij} = (e_i + e_j)/\sqrt{2}$.
Let $\theta\in(0,1)$ and define the heavy index set $J:=\{i: P_{ii}\ge \theta^2\}$.
Then:
\begin{enumerate}
\item[(a)] $|J|\le r/\theta^2$ (since $\sum_i P_{ii}=\mathrm{tr}(P)=r$).
\item[(b)] Let $t:=\frac{1-\theta^2}{2}-\varepsilon^2 \geq 0$. Then all but at most $(r/t^2)$ edges of $F$ have both endpoints in $J$. Consequently,
\[
\bigl|\{\{i,j\}\in F:\ i,j\in J\}\bigr|\ \ge\ |F|\,-\,\frac{r}{t^2}.
\]
In particular, the coordinate subspace
\(
S_J:=\mathrm{span}\{e_k:\ k\in J\}
\)
has $\dim(S_J)=|J|\le r/\theta^2$ and contains exactly all edges in $F\cap E(J,J)$.
\end{enumerate}
\end{lemma}

\begin{proof}
Part (a) is immediate from $\sum_i P_{ii}=r$.

For (b), take any edge $\{i,j\}\in F$ with $i\notin J$ (so $P_{ii}<\theta^2$). Since
\[
\|(I-P)v_{ij}\|_2^2\ =\ 1-\|Pv_{ij}\|_2^2
\ =\ 1-\frac{1}{2}(P_{ii}+P_{jj}+2P_{ij})\ \le\ \varepsilon^2,
\]
we obtain
\[
P_{ii}+P_{jj}+2P_{ij}\ \ge\ 2(1-\varepsilon^2).
\]
Using $P_{ii}<\theta^2$ and $P_{jj}\le 1$,
\[
2P_{ij}\ \ge\ 2(1-\varepsilon^2)-P_{ii}-P_{jj}\ \ge\ 2(1-\varepsilon^2)-(1+\theta^2),
\]
hence $P_{ij}\ge t:=\frac{1-\theta^2}{2}-\varepsilon^2$. For this fixed $i\notin J$, the number of neighbors $j$ with $P_{ij}\ge t$ is at most
\[
\frac{1}{t^2}\sum_{j=1}^d P_{ij}^2\ =\ \frac{P_{ii}}{t^2},
\]
where the equality is due to $P^2 = P$ for any projection matrix. 
Summing over all $i\notin J$ bounds the number of edges in $F$ with at least one light endpoint:
\[
\bigl|F\setminus E(J,J)\bigr|\ \le\ \sum_{i\notin J}\frac{P_{ii}}{t^2}
\ \le\ \frac{1}{t^2}\sum_{i=1}^d P_{ii}\ =\ \frac{r}{t^2}.
\]
The remaining edges in $F$ lie in $E(J,J)$, and each such vector $v_{ij}$ is supported in $J$, hence lies exactly in $S_J$.
\end{proof}

\section{Application to Convex Confidence Sets and Conformal Prediction}\label{sec:applications}
\newcommand{\Ecal}{\mathcal{E}}
\newcommand{\Ecalbeta}{\mathcal{E}_\beta}

\anote{Statement about the Ellipsoidal confidence seet+ VC dim}

\subsection{Convex confidence sets}
\anote{Convex confidence sets}

Our routine for approximating the optimal confidence ellipsoid of a given aspect ratio \(\beta\), naturally gives a guarantee for approximating the optimal \emph{convex confidence set} of a distribution.  This is because any convex set is well-approximated by an ellipsoid (via John's theorem), and indeed any well-conditioned convex set is well-approximated by a well-conditioned ellipsoid~\cite{john1948extremum}.  

This highlights an interesting aspect of the problem of learning confidence sets.  We know that for a family \(\mathcal{F}\) of confidence sets, \(\mathcal{F}\) having bounded VC dimension is sufficient to ensure that we can statistically efficiently learn an approximate best confidence set in \(\mathcal{F}\).  However, bounded VC dimension is not a necessary condition!  Indeed the family of convex sets (and even the family of well-conditioned convex sets) has infinite VC-dimension, even in \(2\) dimensions.  

Bounded VC-dimension is not necessary to achieve coverage, because it is a qualitatively different goal than achieving low binary classification error.  To achieve low binary classification error, it is important to include the positive examples in the set, and exclude the negative examples from your set.  However, to achieve coverage, it is only important that enough points are included in the set, and there is no penalty for including extra points.  Another way of seeing this is that the convex set that we are trying to approximate may have significantly different coverage on the population than on the empirical sample.  However, it is well-approximated by an ellipsoid, and that ellipsoid must have similar coverage on the population and the sample via uniform convergence.

\begin{definition}[Aspect ratio of a convex body]\label{def:aspect-ratio}
Let $X \subset \R^d$ be a convex body with nonempty interior.
Denote by
\[
r(X)\ :=\ \max_{x \in \R^d}\ \sup\{\, r > 0 : B(x, r) \subseteq X \,\},
\qquad
R(X)\ :=\ \min_{x \in \R^d}\ \inf\{\, R > 0 : X \subseteq B(x, R) \,\},
\]
the \emph{inradius} and \emph{circumradius} of $X$, respectively,
where $B(x, r)$ denotes the Euclidean ball of radius $r$ centered at $x$.
The \emph{aspect ratio} of $X$ is
\[
\kappa(X)\ :=\ \frac{R(X)}{r(X)} \ \ge\ 1.
\]
\end{definition}

\begin{theorem}[Competing with convex confidence sets]\label{thm:convex-via-john}
Let $A=\{a_1,\dots,a_n\}\subset\R^d$. 
Fix parameters $\alpha\in(0,1)$ and $\gamma \in (0,1)$. 
There exists a universal constant $c>0$ and a polynomial-time algorithm that outputs an ellipsoid
$\widehat{E}\subset\R^d$ with the following property:

For the minimum volume convex set $K^\star\subset\R^d$ with $|K^\star\cap A|\ge (1-\alpha)n$, let $\kappa(K^*)$ be the aspect ratio of $K^*$.
Then the output $\widehat{E}$ satisfies
\begin{equation}\label{eq:convex-via-john:vol}
\vol(\widehat{E})^{1/d} \;\le\; t_d^{1+\gamma}\,\kappa(K^*)^{\gamma}\,\vol(K^\star)^{1/d},
\qquad\text{and}\qquad
|\widehat{E}\cap A| \;\ge\; \Big(1-\frac{c\,\alpha}{\gamma}\Big)n,
\end{equation}
where $t_d = \sqrt{d}$ if $K^*$ is centrally symmetric; otherwise $t_d = d$.
\end{theorem}

\begin{proof}
Let $E_J$ be the John ellipsoid of $K^\star$. By John's theorem, $E_J\subseteq K^\star \subseteq t_d\,E_J$
with $t_d=d$ in general and $t_d=\sqrt d$ if $K^\star$ is centrally symmetric. Set $E' = t_d\,E_J$.
Then $E'$ has the same condition number as $E_J$ (scaling does not change condition number), and
$K^\star\subseteq E'$, hence $|E'\cap A|\ge |K^\star\cap A|\ge(1-\alpha)n$.

Apply Theorem~\ref{thm:main} to the inlier witness ellipsoid $E'$. Let $\beta_J$ be the condition number of the ellipsoid $E_J$.  Since $E'$ is $\beta_J$-conditioned and
covers at least $(1-\alpha)n$ points, Theorem~\ref{thm:main} yields an ellipsoid $\widehat{E}$ with
\[
\vol(\widehat{E})^{1/d}\;\le\; \beta_J^{\gamma}\,\vol(E')^{1/d}
\;=\; \beta_J^{\gamma}\,t_d\,\vol(E_J)^{1/d}
\;\le\; \beta_J^{\gamma}\,t_d\,\vol(K^\star)^{1/d},
\]
using $\vol(E_J)\le \vol(K^\star)$ because $E_J\subseteq K^\star$. The inlier bound
$|\widehat{E}\cap A| \ge \big(1-\frac{c\alpha}{\gamma}\big)n$ is inherited directly from
Theorem~\ref{thm:main}. 

Let $\lambda_{\min}$ and $\lambda_{\max}$ be the minimum and maximum eigenvalues of the shape matrix of $E_J$, respectively. Then, $\sqrt{\lambda_{\min}}$ and $\sqrt{\lambda_{\max}}$ are the minimum and maximum axes lengths of $E_J$. Since $E_J \subseteq K^* \subseteq t_d E_J$, we have $R(K^*) \geq \sqrt{\lambda_{\max}}$ and $t_d \sqrt{\lambda_{\min}} \geq r(K^*)$. Thus, we have 
$$
\beta_J = \sqrt{\frac{\lambda_{\max}}{\lambda_{\min}}} \leq t_d \frac{R(K^*)}{r(K^*)} = t_d \kappa(K^*).
$$

\end{proof}

\subsection{Conformal Prediction}
\anote{Conformal prediction theorem}

A natural application of learning confidence sets is to conformal prediction.  In conformal prediction, we want to construct a prediction set that covers an unseen test point with probability at least \(1 - \alpha\), assuming only that the test point and the training (empirical) samples are \emph{exchangeable}.  Our algorithm for learning confidence ellipsoids directly implies an algorithm for conformal prediction as follows.

\begin{theorem}[Conformal Prediction with Approximate Volume Optimality]\label{thm:conformal-ellipsoid}
Fix parameters $\alpha\in(0,1)$, $\gamma \in (0,1)$, 
and a condition number upper bound $\beta\ge 1$.
Let $\Ecal_\beta$ denote the class of all origin-centered ellipsoids in $\R^d$ whose condition number is at most $\beta$.
There is a polynomial-time conformal prediction algorithm that, given $Y_1,\ldots,Y_n\in\R^d$, outputs a prediction set $\widehat{C}\subseteq\R^d$ such that:

\begin{enumerate}[(a)]
\item \textbf{Marginal coverage.} If $Y_1,\ldots,Y_{n+1}$ are exchangeable, then
\[
\Pr\big[Y_{n+1}\in \widehat{C}\big]\ \ge\ 1-\alpha.
\]

\item \textbf{Approximate volume optimality versus $\beta$-conditioned ellipsoids.}
If $Y_1,\ldots,Y_{n+1}$ are i.i.d.\ from an arbitrary distribution $\mathcal{D}$ on $\R^d$, then for any $\delta\in(0,1)$, whenever $n = \Omega(d^2/\gamma^2)$,
the following holds with probability at least $1-\delta$ over the training sample:
\[
\vol(\widehat{C})^{1/d}\ \le\ \beta^{\gamma}\,\vol\!\big(C^\star_{\beta}\big)^{1/d},
\]
where $C^\star_{\beta}$ is the optimal $\beta$-conditioned ellipsoid achieving a slightly stronger coverage constraint, such that for a universal constant $c > 0$
\[
C^\star_{\beta}\ \in\ \arg\min_{C\in \Ecal_\beta}\ \vol(C)
\quad\text{s.t.}\quad
\Pr_{Y\sim\mathcal{D}}[Y\in C]\ \ge\ 1-c\alpha\gamma.
\]
\end{enumerate}
\end{theorem}

\begin{proof}
    Suppose data points $Y_1, \dots, Y_{n+1}$ are drawn i.i.d. from an arbitrary distribution $\mathcal{D}$. By Corollary~\ref{cor:main}, since $n = \Omega(d^2/ \gamma^2)$ and $\gamma < \alpha$, our algorithm  finds a set $\widehat{C}$ such that
    \[
    \vol(\widehat{C})^{1/d}\ \le\ \beta^{\gamma}\,\vol\!\big(C^\star_{\beta}\big)^{1/d},
    \]
    where $C^\star_{\beta}$ is an $\beta$-conditioned ellipsoid satisfies
    \[
    C^\star_{\beta}\ \in\ \arg\min_{C\in \Ecal_\beta}\ \vol(C)
    \quad\text{s.t.}\quad
    \Pr_{Y\sim\mathcal{D}}[Y\in C]\ \ge\ 1-c\alpha\gamma.
    \]
    Moreover, we have 
    $$
    \Pr[Y_{n+1} \in \widehat{C}] = \Pr_{Y\sim\mathcal{D}}[Y\in \widehat{C}] \geq 1- O\left(\frac{c\alpha \gamma}{\gamma} - \sqrt{\frac{d^2}{n}}\right) \geq 1-\alpha.
    $$

   We then construct a conformal prediction by constructing a nested system as in~\cite{gao2025confidence}. This conformal predictor satisfies the margin coverage under exchangeable data.
\end{proof}


\bibliographystyle{alpha}
\bibliography{ref}

@inproceedings{lai2016agnostic,
  title={Agnostic estimation of mean and covariance},
  author={Lai, Kevin A and Rao, Anup B and Vempala, Santosh},
  booktitle={2016 IEEE 57th Annual Symposium on Foundations of Computer Science (FOCS)},
  pages={665--674},
  year={2016},
  organization={IEEE}
}

@inproceedings{diakonikolas2016robust,
  title={Robust Estimators in High Dimensions without the Computational Intractability},
  author={Diakonikolas, Ilias and Kamath, Gautam and Kane, Daniel M and Li, Jerry and Moitra, Ankur and Stewart, Alistair},
  booktitle={2016 IEEE 57th Annual Symposium on Foundations of Computer Science (FOCS)},
  pages={655--664},
  year={2016},
  organization={IEEE}
}

@article{ToddY2007,
author = {Todd, Michael J. and Y\i{}ld\i{}r\i{}m, E. Alper},
title = {On Khachiyan's algorithm for the computation of minimum-volume enclosing ellipsoids},
year = {2007},
issue_date = {August, 2007},
publisher = {Elsevier Science Publishers B. V.},
address = {NLD},
volume = {155},
number = {13},
issn = {0166-218X},
url = {https://doi.org/10.1016/j.dam.2007.02.013},
doi = {10.1016/j.dam.2007.02.013},
journal = {Discrete Appl. Math.},
month = aug,
pages = {1731–1744},
numpages = {14}
}

@article{SunF2004,
author = {Sun, Peng and Freund, Robert M.},
title = {Computation of Minimum-Volume Covering Ellipsoids},
year = {2004},
issue_date = {Sep. - Oct. 2004},
publisher = {INFORMS},
address = {Linthicum, MD, USA},
volume = {52},
number = {5},
issn = {0030-364X},
url = {https://doi.org/10.1287/opre.1040.0115},
doi = {10.1287/opre.1040.0115},
journal = {Oper. Res.},
month = oct,
pages = {690–706},
numpages = {17}
}

@inproceedings{badoiu2002approximate,
author = {Badoiu, Mihai and Har-Peled, Sariel and Indyk, Piotr},
title = {Approximate clustering via core-sets},
year = {2002},
isbn = {1581134959},
publisher = {Association for Computing Machinery},
address = {New York, NY, USA},
url = {https://doi.org/10.1145/509907.509947},
doi = {10.1145/509907.509947},
booktitle = {Proceedings of the Thiry-Fourth Annual ACM Symposium on Theory of Computing},
pages = {250–257},
numpages = {8},
location = {Montreal, Quebec, Canada},
series = {STOC '02}
}

@book{huber2004robust,
  title={Robust Statistics},
  author={Huber, P.J.},
  isbn={9780471650720},
  lccn={2004266770},
  series={Wiley Series in Probability and Statistics - Applied Probability and Statistics Section Series},
  url={https://books.google.com/books?id=e62RhdqIdMkC},
  year={2004},
  publisher={Wiley}
}

@inproceedings{khot2003UGC,
author = {Khot, Subhash},
title = {On the power of unique 2-prover 1-round games},
year = {2002},
isbn = {1581134959},
publisher = {Association for Computing Machinery},
address = {New York, NY, USA},
url = {https://doi.org/10.1145/509907.510017},
doi = {10.1145/509907.510017},
booktitle = {Proceedings of the Thiry-Fourth Annual ACM Symposium on Theory of Computing},
pages = {767–775},
numpages = {9},
location = {Montreal, Quebec, Canada},
series = {STOC '02}
}

@inproceedings{raghavendra2010graph,
author = {Raghavendra, Prasad and Steurer, David},
title = {Graph expansion and the unique games conjecture},
year = {2010},
isbn = {9781450300506},
publisher = {Association for Computing Machinery},
address = {New York, NY, USA},
url = {https://doi.org/10.1145/1806689.1806792},
doi = {10.1145/1806689.1806792},
booktitle = {Proceedings of the Forty-Second ACM Symposium on Theory of Computing},
pages = {755–764},
numpages = {10},
location = {Cambridge, Massachusetts, USA},
series = {STOC '10}
}

@InProceedings{HardtMoitra13,
  title = 	 {Algorithms and Hardness for Robust Subspace Recovery},
  author = 	 {Hardt, Moritz and Moitra, Ankur},
  booktitle = 	 {Proceedings of the 26th Annual Conference on Learning Theory},
  pages = 	 {354--375},
  year = 	 {2013},
  editor = 	 {Shalev-Shwartz, Shai and Steinwart, Ingo},
  volume = 	 {30},
  series = 	 {Proceedings of Machine Learning Research},
  address = 	 {Princeton, NJ, USA},
  month = 	 {12--14 Jun},
  publisher =    {PMLR},
  url = 	 {https://proceedings.mlr.press/v30/Hardt13.html}
}

@article{khachiyan,
author = {Khachiyan, L.},
title = {On the Complexity of Approximating Extremal Determinants in Matrices},
year = {1995},
issue_date = {March 1995},
publisher = {Academic Press, Inc.},
address = {USA},
volume = {11},
number = {1},
issn = {0885-064X},
url = {https://doi.org/10.1006/jcom.1995.1005},
doi = {10.1006/jcom.1995.1005},
journal = {J. Complex.},
month = mar,
pages = {138–153},
numpages = {16}
}

@article{lei2013distribution,
	author = {Jing Lei and James Robins and Larry Wasserman},
	journal = {Journal of the American Statistical Association},
	number = {501},
	pages = {278-287},
	title = {Distribution-Free Prediction Sets},
	volume = {108},
	year = {2013}}

@article{Ball1991Volume,
  author = {K. Ball},
  title = {Volume ratios and a reverse isoperimetric inequality},
  journal = {Journal of the London Mathematical Society},
  volume = {44},
  issue = {2},
  pages = {351--359},
  year = {1991},
  doi = {10.1112/jlms/s2-44.2.351}
}

@incollection{Ball1989VolumesOS,
  title={Volumes of sections of cubes and related problems},
  author={Keith Ball},
  booktitle={Geometric Aspects of Functional Analysis},
  pages={251-260},
  year={1989},
  publisher={Springer, Berlin, Heidelberg},
  doi={10.1007/bfb0090058}
}

@InProceedings{cohen19a,
  title = 	 {A near-optimal algorithm for approximating the John Ellipsoid},
  author =       {Cohen, Michael B. and Cousins, Ben and Lee, Yin Tat and Yang, Xin},
  booktitle = 	 {Proceedings of the Thirty-Second Conference on Learning Theory},
  pages = 	 {849--873},
  year = 	 {2019},
  editor = 	 {Beygelzimer, Alina and Hsu, Daniel},
  volume = 	 {99},
  series = 	 {Proceedings of Machine Learning Research},
  month = 	 {25--28 Jun},
  publisher =    {PMLR},
  url = 	 {https://proceedings.mlr.press/v99/cohen19a.html}
}

@inproceedings{gao2025volume,
  author    = {Gao, Chao and Shan, Liren and Srinivas, Vaidehi and Vijayaraghavan, Aravindan},
  title     = {Volume Optimality in Conformal Prediction with Structured Prediction Sets},
  booktitle = {Proceedings of the 42nd International Conference on Machine Learning (ICML)},
  year      = {2025},
  series    = {ICML'25},
  location  = {Vancouver, Canada},
  url       = {https://arxiv.org/abs/2502.16658}
}

@inproceedings{gao2025confidence,
  author    = {Gao, Chao and Shan, Liren and Srinivas, Vaidehi and Vijayaraghavan, Aravindan},
  title     = {Computing High-Dimensional Confidence Sets for Arbitrary Distributions},
  booktitle = {Proceedings of the 38th Annual Conference on Learning Theory (COLT)},
  series    = {COLT'25},
  year      = {2025},
  url       = {https://arxiv.org/abs/2504.02723}
}

@inproceedings{srinivas2025online,
  author    = {Srinivas, Vaidehi},
  title     = {Online Conformal Prediction with Efficiency Guarantees},
  booktitle = {Proceedings of the 2026 Annual ACM-SIAM Symposium on Discrete Algorithms (SODA)},
chapter = {},
  year      = {2026},
  url       = {https://arxiv.org/abs/2507.02496}
}

@inproceedings{BCPV,
  title={Smoothed Analysis in Unsupervised Learning via Decoupling},
  author={Bhaskara, Aditya and Chen, Aidao and Perreault, Aidan and Vijayaraghavan, Aravindan},
  booktitle={Proceedings of the 60th Annual IEEE Symposium on Foundations of Computer Science (FOCS)},
  year={2019},
  organization={IEEE}
}

@inproceedings{bakshi2021list,
  title={List-decodable subspace recovery: Dimension independent error in polynomial time},
  author={Bakshi, Ainesh and Kothari, Pravesh K},
  booktitle={Proceedings of the 2021 ACM-SIAM Symposium on Discrete Algorithms (SODA)},
  pages={1279--1297},
  year={2021},
  organization={SIAM}
}

@book{Bhatia,
  author = {Bhatia, Rajendra},
  isbn = {0387948465},
  publisher = {Springer},
  title = {Matrix Analysis},
  volume = 169,
  year = 1997
}

@Book{matrix,
author = { Horn, Roger A. and Johnson, Charles R. },
title = { Matrix analysis / Roger A. Horn, Charles R. Johnson },
isbn = { 0521305861 0521386322 },
publisher = { Cambridge University Press, Cambridge [Cambridgeshire] ; New York },
pages = { xiii, 561 p. },
year = { 1985 },
type = { Book },
url = { http://www.loc.gov/catdir/toc/cam023/85007736.html },
language = { English },
subjects = { Matrices. },
life-dates = { 1985 -  },
catalogue-url = { http://nla.gov.au/nla.cat-vn33841 },
}

@book{sdp,
       author="Martin Gr{\"{o}}tschel and L{\'{a}}szl{\'{o}} Lov{\'{a}}sz and Alexander Schrijver",
      title="Geometric Algorithms and Combinatorial Optimization",
      publisher="Springer-Verlag",
      address="New York",
      year="1988"}

@article{Sadinle2016LeastAS,
  title={Least Ambiguous Set-Valued Classifiers With Bounded Error Levels},
  author={Mauricio Sadinle and Jing Lei and Larry A. Wasserman},
  journal={Journal of the American Statistical Association},
  year={2016},
  volume={114},
  pages={223 - 234},
  url={https://api.semanticscholar.org/CorpusID:622583}
}

@incollection{john1948extremum,
  address = {New York},
  author = {John, Fritz},
  booktitle = {Studies and Essays : Courant Anniversary Volume},
  pages = {187-204},
  publisher = {John Wiley \& Sons : Interscience Division},
  title = {Extremum problems with inequalities as subsidiary conditions},
  year = 1948
}

@article{braun2025volume,
  author={Sacha Braun and Liviu Aolaritei and Michael I. Jordan and Francis R. Bach},
  title={Minimum Volume Conformal Sets for Multivariate Regression},
  year={2025},
  month={March},
  cdate={1740787200000},
  journal={CoRR},
  volume={abs/2503.19068},
  url={https://doi.org/10.48550/arXiv.2503.19068}
}

@InProceedings{conformalellipsoid2024xu,
  title = 	 {Conformal prediction for multi-dimensional time series by ellipsoidal sets},
  author =       {Xu, Chen and Jiang, Hanyang and Xie, Yao},
  booktitle = 	 {Proceedings of the 41st International Conference on Machine Learning},
  pages = 	 {55076--55099},
  year = 	 {2024},
  editor = 	 {Salakhutdinov, Ruslan and Kolter, Zico and Heller, Katherine and Weller, Adrian and Oliver, Nuria and Scarlett, Jonathan and Berkenkamp, Felix},
  volume = 	 {235},
  series = 	 {Proceedings of Machine Learning Research},
  month = 	 {21--27 Jul},
  publisher =    {PMLR},
  url = 	 {https://proceedings.mlr.press/v235/xu24m.html}
}

@InProceedings{conformalellipsoid2022,
  title = 	 {Ellipsoidal conformal inference for Multi-Target Regression},
  author =       {Messoudi, Soundouss and Destercke, S\'{e}bastien and Rousseau, Sylvain},
  booktitle = 	 {Proceedings of the Eleventh Symposium on Conformal and Probabilistic Prediction with Applications},
  pages = 	 {294--306},
  year = 	 {2022},
  editor = 	 {Johansson, Ulf and Boström, Henrik and An Nguyen, Khuong and Luo, Zhiyuan and Carlsson, Lars},
  volume = 	 {179},
  series = 	 {Proceedings of Machine Learning Research},
  month = 	 {24--26 Aug},
  publisher =    {PMLR},
  url = 	 {https://proceedings.mlr.press/v179/messoudi22a.html}
}

@inproceedings{nikolov2015max,
author = {Nikolov, Aleksandar},
title = {Randomized Rounding for the Largest Simplex Problem},
year = {2015},
isbn = {9781450335362},
publisher = {Association for Computing Machinery},
address = {New York, NY, USA},
url = {https://doi.org/10.1145/2746539.2746628},
doi = {10.1145/2746539.2746628},
booktitle = {Proceedings of the Forty-Seventh Annual ACM Symposium on Theory of Computing},
pages = {861–870},
numpages = {10},
location = {Portland, Oregon, USA},
series = {STOC '15}
}

@inproceedings{singh2016maxdet,
author = {Nikolov, Aleksandar and Singh, Mohit},
title = {Maximizing determinants under partition constraints},
year = {2016},
isbn = {9781450341325},
publisher = {Association for Computing Machinery},
address = {New York, NY, USA},
url = {https://doi.org/10.1145/2897518.2897649},
doi = {10.1145/2897518.2897649},
booktitle = {Proceedings of the Forty-Eighth Annual ACM Symposium on Theory of Computing},
pages = {192–201},
numpages = {10},
location = {Cambridge, MA, USA},
series = {STOC '16}
}

@book{ars_book, 
    place={Cambridge}, 
    title={Algorithmic High-Dimensional Robust Statistics}, 
    publisher={Cambridge University Press}, 
    author={Diakonikolas, Ilias and Kane, Daniel M.}, 
    year={2023}
}

@book{boyd2004convex,
  title={Convex optimization},
  author={Boyd, Stephen and Vandenberghe, Lieven},
  year={2004},
  publisher={Cambridge university press}
}

@article{scott2005learning,
  title={Learning minimum volume sets},
  author={Scott, Clayton and Nowak, Robert},
  journal={Advances in neural information processing systems},
  volume={18},
  year={2005}
}

@article{garcia2003level,
  title={Level sets and minimum volume sets of probability density functions},
  author={Garcia, Javier Nu{\~n}ez and Kutalik, Zoltan and Cho, Kwang-Hyun and Wolkenhauer, Olaf},
  journal={International journal of approximate reasoning},
  volume={34},
  number={1},
  pages={25--47},
  year={2003},
  publisher={Elsevier}
}

@article{rousseeuw1985multivariate,
  title={Multivariate estimation with high breakdown point},
  author={Rousseeuw, P},
  journal={Mathematical Statistics and Applications B},
  year={1985}
}

@article{van2009minimum,
  title={Minimum volume ellipsoid},
  author={Van Aelst, Stefan and Rousseeuw, Peter},
  journal={Wiley Interdisciplinary Reviews: Computational Statistics},
  volume={1},
  number={1},
  pages={71--82},
  year={2009},
  publisher={Wiley Online Library}
}

@article{polonik1999concentration,
  title={Concentration and goodness-of-fit in higher dimensions:(Asymptotically) distribution-free methods},
  author={Polonik, Wolfgang},
  journal={The Annals of Statistics},
  volume={27},
  number={4},
  pages={1210--1229},
  year={1999},
  publisher={Institute of Mathematical Statistics}
}

@article{einmahl1992generalized,
  title={Generalized quantile processes},
  author={Einmahl, John HJ and Mason, David M},
  journal={The Annals of Statistics},
  pages={1062--1078},
  year={1992},
  publisher={JSTOR}
}

@article{polonik1997minimum,
  title={Minimum volume sets and generalized quantile processes},
  author={Polonik, Wolfgang},
  journal={Stochastic processes and their applications},
  volume={69},
  number={1},
  pages={1--24},
  year={1997},
  publisher={Elsevier}
}

@article{Barthe1997,
  title={On a reverse form of the Brascamp-Lieb inequality},
  author={Franck Barthe},
  journal={Inventiones mathematicae},
  year={1997},
  volume={134},
  pages={335-361},
  url={https://api.semanticscholar.org/CorpusID:14967794}
}

@book{Toddbook,
author = {Todd, Michael J.},
title = {Minimum-Volume Ellipsoids: Theory and Algorithms},
year = {2016},
isbn = {1611974372},
publisher = {SIAM-Society for Industrial and Applied Mathematics},
address = {Philadelphia, PA, USA}
}

@inproceedings{AhmadiML2014,
  author    = {A. A. Ahmadi and D. Malioutov and R. Luss},
  title     = {Robust Minimum Volume Ellipsoids and Higher-Order Polynomial Level Sets},
  booktitle = {NIPS Workshop on Optimization for Machine Learning},
  year      = {2014},
  address   = {Montr{\'e}al, Qu{\'e}bec, Canada}
}

@article{CalafioreG04,
  author={Giuseppe Carlo Calafiore and Laurent El Ghaoui},
  title={Ellipsoidal bounds for uncertain linear equations and dynamical systems},
  year={2004},
  cdate={1072915200000},
  journal={Autom.},
  volume={40},
  number={5},
  pages={773-787},
  url={https://doi.org/10.1016/j.automatica.2004.01.001}
}

@inproceedings{Agullo1996,
  author    = {Jos{\'e} A. Agull{\'o} Candela},
  title     = {Exact Iterative Computation of the Multivariate Minimum Volume Ellipsoid Estimator with a Branch and Bound Algorithm},
  booktitle = {COMPSTAT},
  editor    = {Alfredo Prat},
  year      = {1996},
  pages     = {175--180},
  publisher = {Physica-Verlag HD},
  address   = {Heidelberg},
  doi       = {10.1007/978-3-642-46992-3_16}
}

@article{Ahipasaoglu2014,
  author    = {Sinem Ahipa{\c{s}}ao{\u{g}}lu},
  title     = {Fast Algorithms for the Minimum Volume Estimator},
  journal   = {Journal of Global Optimization},
  volume    = {59},
  number    = {2--3},
  pages     = {351--376},
  year      = {2014},
  month     = {June},
  doi       = {10.1007/s10898-013-0116-5}
}

@misc{henderson2025adaptiveinferencerandomellipsoids,
      title={Adaptive inference with random ellipsoids through Conformal Conditional Linear Expectation}, 
      author={Iain Henderson and Adrien Mazoyer and Fabrice Gamboa},
      year={2025},
      eprint={2409.18508},
      archivePrefix={arXiv},
      primaryClass={math.ST},
      url={https://arxiv.org/abs/2409.18508}, 
}

@misc{lindemann2024,
      title={Conformal Prediction for Distribution-free Optimal Control of Linear Stochastic Systems}, 
      author={Eleftherios E. Vlahakis and Lars Lindemann and Pantelis Sopasakis and Dimos V. Dimarogonas},
      year={2024},
      eprint={2411.19132},
      archivePrefix={arXiv},
      primaryClass={eess.SY},
      url={https://arxiv.org/abs/2411.19132}, 
}

@TechReport{thorsten,
type={Technical Reports},
institution={Technische UniversitÃ¤t Dortmund, Sonderforschungsbereich 475: KomplexitÃ¤tsreduktion in multivariaten Datenstrukturen},
author={Bernholt, Thorsten},
title={Robust Estimators are Hard to Compute},
year={2006},
number={2005,52},
doi={None},
url={https://ideas.repec.org/p/zbw/sfb475/200552.html},
}

@ARTICLE{subspacerecoverySurvey,
  author={Lerman, Gilad and Maunu, Tyler},
  journal={Proceedings of the IEEE}, 
  title={An Overview of Robust Subspace Recovery}, 
  year={2018},
  volume={106},
  number={8},
  pages={1380-1410},
  doi={10.1109/JPROC.2018.2853141}}

@misc{hopkins2020pointlocationactivelearning,
      title={Point Location and Active Learning: Learning Halfspaces Almost Optimally}, 
      author={Max Hopkins and Daniel M. Kane and Shachar Lovett and Gaurav Mahajan},
      year={2020},
      eprint={2004.11380},
      archivePrefix={arXiv},
      primaryClass={cs.CG},
      url={https://arxiv.org/abs/2004.11380}, 
}

@article{Diakonikolas2022ASP,
  title={A Strongly Polynomial Algorithm for Approximate Forster Transforms and Its Application to Halfspace Learning},
  author={Ilias Diakonikolas and Christos Tzamos and Daniel M. Kane},
  journal={Proceedings of the 55th Annual ACM Symposium on Theory of Computing},
  year={2022},
  url={https://api.semanticscholar.org/CorpusID:254275040}
}

\appendix

\section{Computational Hardness of Obtaining Non-trivial Volume Approximation Guarantees} \label{app:nphardness}

In this section, we show evidence of computational intractability for obtaining any non-trivial approximations for the volume of the ellipsoid. In we will show that it is computationally hard to determine whether there exists a $1-\alpha$-fraction of points that lie on a subspace. Note that the volume of the minimum volume ellipsoid containing a set of points $S$ is zero exactly when $S$ belongs to a subspace of dimension at most $d-1$. 

We start with a simple NP-hardness result due to~\cite{khachiyan}. 

\begin{proposition}\label{prop:np-hard:khachiyan}
There exists a constant $\alpha_0>0$, such that for any $\alpha \ge \alpha_0(1-\tfrac{1}{d})+\frac{1}{d}$ given any set of $m$ points in $d$ dimensions, assuming $P \ne NP$, there is no algorithm running in polynomial time (in $m,d$) that determines if there exists an ellipsoid of volume $0$ that contains at least $(1-\alpha) m$ points. 
\end{proposition}
\begin{proof}
Khachiyan shows that it is NP-hard to determine if there exists a $d-1$ dimensional subspace containing at least $(1-\alpha_0)(1-1/d)$ fraction of points~\cite{khachiyan}. The volume of the minimum volume ellipsoid containing a set of points $S$ is zero if and only if $S$ belongs to a subspace of dimension at most $d-1$. Moreover, we can add additional random points in $\mathbb{R}^d$ to get the same hardness for any coverage $1-\alpha \le (1-\alpha_0)(1-1/d)$. 
\end{proof}

\section{Proofs in Section~\ref{sec:SSE-hard}}\label{apx:SSE}

\begin{proof}[Proof of Claim~\ref{clm:span}]
Let $S := N_B(F) \subseteq V$ be the set of vertices incident to edges in $F$, and let
$H := (S,F)$ be the subgraph of $G$ induced by $F$ on the vertex set $S$.

We first prove the second inequality. 
Since for any edge $e = (i,j) \in F$, the vector $u_e' = (e_i+e_j)/\sqrt{2}$, all vectors in $P$ lie in the coordinate subspace spanned by $e_i$ for all $i \in S$. Therefore
\[
\dim(\mathrm{span}\,P) \;\le\; |S| \;=\; |N_B(F)|.
\]

We now prove the first inequality. 
Let the connected components of $H$ be $H_1,\dots,H_r$ with vertex sets
$S_1,\dots,S_r$ (a partition of $S$) and edge sets $F_1,\dots,F_r$ (a partition of $F$).
Since the supports of the vectors $\{u'_e : e\in F_c\}$ lie in disjoint coordinate blocks
$S_c$, we have
\[
\dim(\mathrm{span}\,P) \;=\; \sum_{c=1}^r \dim\!\big(\mathrm{span}\,\{u_e : e\in F_c\}\big).
\]
Fix a component $H_c$. Choose a spanning tree $T_c \subseteq F_c$; then $|T_c|=|S_c|-1$.
We claim that the vectors $\{u'_e : e \in T_c\}$ are linearly independent. Indeed, suppose
$\sum_{e\in T_c} \lambda_e\, u'_e = 0$. Pick a leaf $v$ of the tree $T_c$, and let
$e_v=\{v,w\}$ be its unique incident edge in $T_c$. Looking at the $v$-th coordinate of the
preceding relation yields $\lambda_{e_v}=0$, because $u'_{e_v}$ has one in coordinate $v$ while all
other $u'_e$ for $e\in T_c \setminus \{e_v\}$ have zero in that coordinate. Removing $e_v$ and $v$ from the tree and repeating this
argument inductively forces all coefficients to be zero. Thus
\[
\dim\!\big(\mathrm{span}\,\{u_e : e\in F_c\}\big) = |T_c| \;=\; |S_c|-1.
\]
Summing over components gives
\[
\dim(\mathrm{span}\,P) = \sum_{c=1}^r (|S_c|-1) \;=\; |S| - r.
\]
Finally, since $S$ consists of vertices incident to edges of $F$, each component $H_c$
has at least two vertices, so $r \le |S|/2$. Therefore
\[
\dim(\mathrm{span}\,P) \;\ge\; |S| - r \;\ge\; \frac{|S|}{2}
\;=\; \frac{|N_B(F)|}{2}.
\]
Combining with the upper bound proves the claim.
\end{proof}

\section{Algorithm Details}
\label{app:algo-details}

\subsection{Preliminaries}

The Schur complement is a useful way to compute the determinant and eigenvalues of a block matrix. 

\begin{fact}[Schur complement~\cite{Bhatia}]\label{fact:schur}
Consider a $(p+q) \times (p+q)$ block matrix $M$ where the submatrix $X \in \R^{p \times p}$ is invertible. 
\[
M= \begin{bmatrix}
\begin{array}{c|c}
X & Y \\ \hline
W & Z
\end{array}
\end{bmatrix}
\]
Then the Schur complement for the block $X$ is given by $Z - W X^{-1} Y$, and
\begin{align*}
\det(M)= \det(X) \cdot \det(Z - WX^{-1} Y). 
\end{align*}
\end{fact}

\subsection{Bounding ball}
\label{sec:bounding-ball}

The first step of our algorithm is to restrict our attention to a bounding ball that has radius proportional to the longest axis of the optimal ellipsoid.  This can be done via a simple search procedure.  

\begin{lemma}[Bounding ball]\label{lem:bounding-ball}
    For the set of points \(A = \{a_1, \dots, a_n\} \subseteq \mathbb{R}^d\), create a set of balls 
    \[\mathfrak{B} = \{B(a_i, \|a_i - a_j\|) : i, j \in [n] \}.\]
    Remove any \(B \in \mathfrak{B}\) such that \(|B \cap A| \le (1 - \alpha) n\).  We have that \(|\mathfrak{B}| \le n^2\).  
    
    Let \(E^\star\) be the minimum volume ellipsoid of condition number at most \(\beta\) that contains at least \((1 - \alpha)\)-fraction of the points.  Then, there exists a ball \(\widehat{B} \in \mathfrak{B}\) such that \(\widehat{B}\) contains all of the points in \(E^\star\), and the radius of \(\widehat{B}\) is at most twice the longest axis length of \(E^\star\). 
\end{lemma}

\begin{proof}
The bound \(|\mathfrak{B}|\le n^2\) is immediate from the construction since there are at most $n^2$ pairs of $i,j \in [n]$.

Let \(c\) be the center of the ellipsoid \(E^\star\), and let $R = \max_{x \in E^*} \|x-c\|$ be the length of its longest semi-axis. 
Let $a_i, a_j \in A\cap E^\star$ be the pair of centers that maximize the distance among all pairs of points in \(A\cap E^\star\), i.e.,
\[
\|a_i-a_j\|=\max\{\|a-a'\|: a,a'\in A\cap E^\star\}.
\]
Set the ball \(\widehat{B}:=B(a_i,\|a_i-a_j\|)\).

By construction, for every \(a_\ell\in A\cap E^\star\) we have 
\(\|a_\ell-a_i\|\le \|a_j-a_i\|\), so \(A\cap E^\star\subseteq \widehat{B}\); 
that is, \(\widehat{B}\) contains all points of \(A\) that lie in \(E^\star\). Thus, we have the ball $\widehat{B} \in \mathfrak{B}$.
Moreover, since \(a_i,a_j\in E^\star\),
\[
\|a_i-a_j\|\le \|a_i-c\|+\|a_j-c\|\le R+R=2R,
\]
so the radius of \(\widehat{B}\) is at most \(2R\), 
i.e., at most twice the longest semi-axis length of \(E^\star\). 
This proves the claim.
\end{proof}

\subsection{Iterative outlier removal via dual}
\label{sec:outlier-removal}


We analyze the number of points that are removed in our iterative outlier removal process.  We use Chebyshev's inequality with the observation that our removal process forms a martingale. 

\begin{proposition}[Point removal martingale]
    Let the random variable $w_{(nj) + i}$ be the weight that is assigned by the SDP to point $i$ in iteration $j \in [R]$ ($j$th time the SDP is solved) in Algorithm \ref{fig:ellipsoid_tcolor}, where we define $w_{(nj) + i} = 0$ for points that have already been deleted by iteration $j$.

    Let $X_{(nj) + i} \sim \mathrm{Bern}(w_{(nj) + i})$ be the indicator random variable of whether point $i$ was removed in iteration $j$.  
    Define $Z_{(nj) + i}$ as the deviation of $X_{(nj) + i}$ from its expectation,
    \[Z_{(nj) + i} = X_{(nj) + i} - w_{(nj) + i}.\]

    We have that 
    \[\E \left[Z_k | Z_{0}, \dots, Z_{k - 1}  \right] = 0, \]
    \[\E \left[Z_k^2\right] = \E[w_k (1 - w_k)] \le \E[w_k].\]

    Thus, for any subset of the variables $S \subseteq [nR]$, we have 
    \begin{align*}
        \Var \left[ \sum_{k \in S}Z_k\right] &= \E \left[\left( \sum_{k \in S}Z_k \right)^2 \right] 
        = \sum_{k \in S} \sum_{k' \in S} \E\left[ Z_k Z_{k'} \right] \\
        &= \sum_{k \in S} \E \left[ Z_k^2 \right] \le \E \left[\sum_{k \in S} w_k \right] = \E \left[ \sum_{k \in S} X_k \right].
    \end{align*}

    \label{prop:point-removal-martingale}
\end{proposition}

This allows us to conclude the two statements we need about the outlier removal process.  First we argue that over $R$ rounds, we will make progress removing the outliers.

\begin{lemma}[Outlier removal progress]\label{lem:outlier-removal-progress}
    Assume \(\alpha n\) is at least a sufficiently large constant. There exists a constant \(c\) such that if we run the iterative outlier removal procedure for \(R \ge c \frac{\alpha n}{\gamma d}\) iterations, then with probability at least \(0.99\), there exists an iteration on which less than \(\gamma d\) mass was placed on outliers.  
\end{lemma}

\begin{proof}
    Consider the random variables as defined in \Cref{prop:point-removal-martingale}.  Let $S \subseteq [nR]$ be the subset of indices that correspond to outliers, 
    \[S = \{ (nj) + i : a_i \notin E^\star \},\]
    where $E^\star$ is the optimal solution.

    Let 
    \[Z_S = \sum_{k \in S} Z_k\]
    be the deviation of the number of outliers removed from its expectation.  
        
    By \Cref{prop:point-removal-martingale} we have 
    \[\Var (Z_S) \le \E\left[ \sum_{k \in S} X_k \right] \le \alpha n,\]
    since there are at most $\alpha n$ outliers that can be removed in any run of the algorithm.
    By Chebyshev's inequality we have 
    \[\ProbOp \left[ |Z_S| \ge (c - 1) \alpha n \right] \le \frac{\alpha n}{(c - 1)^2(\alpha n)^2} = \frac{1}{(c - 1)^2 \alpha n}.\]

    Let \(\mathcal{E}\) be the event that on all \(R\) iterations, the SDP placed at least $\gamma d$ mass on outliers.  We have that in the event \(\mathcal{E}\),
    \begin{align*}
        \sum_{k \in S} Z_k  &= \sum_{k \in S} X_k - \sum_{k \in S} w_k \\
        &\le \alpha n - R\cdot\gamma d\\
        &\le -\left( c - 1 \right) \alpha n.
    \end{align*}
    Thus $\mathcal{E}$ is a subset of the event that \(|Z_S| \ge (c - 1) \alpha n\), and the probability of event $\mathcal{E}$ is bounded by $\frac{1}{(c - 1)^2\alpha n} $.  Taking $c = 2$ and $\alpha n \ge 100$ proves the claim.  
\end{proof}

We also must show that over $R$ rounds we do not remove too many points total (outliers and inliers together).  

\begin{lemma}[Bound on number of points removed]\label{lem:dont-remove-too-many-points}
    Assume \(\frac{\alpha n}{\gamma}\) is at least a sufficiently large constant.  Let \(c \in \mathbb{R}\) be a constant.  Then there exists a constant \(c_2\) such that, if we run the iterative outlier removal procedure for \(R = c \frac{\alpha n}{\gamma d}\) iterations, then with probability at least \(0.99\), we remove a total of at most 
    \[c_2 \frac{\alpha n}{\gamma} \text{ points}.\]
\end{lemma}

\begin{proof}
    Let 
    \[Z = \sum_{k \in [nR]} Z_k\]
    be the deviation of the number of points removed from its expectation.  By \Cref{prop:point-removal-martingale} we have 
    \[\Var (Z) \le \E \left[\sum_{k \in [nR]} X_k \right] \le Rd,\]
    since the SDP assigns weight $d$ in each round.  By Chebyshev's inequality, we have 
    \[\ProbOp \left[ |Z| \ge c_3 Rd \right] \le \frac{Rd}{c_3^2 R^2 d^2} = \frac{1}{c_3^2 Rd}.\]
    Thus, 
    \[\ProbOp \left[ \sum_{k \in [nR]} X_k \ge c_2 \frac{\alpha n}{\gamma} \right] \leq \ProbOp \left[ \sum_{k \in [nR]} Z_k \ge \left(\frac{c_2}{c} - 1\right)  Rd \right] \le \frac{1}{(c_2/c - 1)^2 Rd}.\]
    Taking $c_2$ such that \(c_2/c - 1 \ge 1\), and assuming $Rd = \frac{\alpha n}{\gamma} \ge 100$ proves the claim.  
\end{proof}

\subsection{Dual equivalence}

\begin{proposition}[Dual equivalence under the lift]
\label{prop:dual-equiv-existence}
Let $a_1,\ldots,a_n\in\mathbb{R}^d$ and $\tilde a_i:=(a_i,1)\in\mathbb{R}^{d+1}$.
Let $w$ be any optimal dual solution for the free-center MVEE (\Cref{prop:dual-MVE-center}) in $\mathbb{R}^d$. 
Then $\wtw= \frac{d+1}{d} w$ is an optimal dual solution for the origin-centered MVEE (\Cref{prop:dual-MVE}) of the lifted points $(\tilde a_i)$ in $\mathbb{R}^{d+1}$.

Conversely, let $\wtw$ be any optimal dual solution for the origin-centered MVEE (\Cref{prop:dual-MVE}) of the lifted points in $\mathbb{R}^{d+1}$. 
Then $w= \frac{d}{d+1} \wtw$ is an optimal dual solution for the free-center MVEE (\Cref{prop:dual-MVE-center}) in $\mathbb{R}^d$.
\end{proposition}

\begin{proof}
We use KKT conditions for both problems.

\medskip\noindent
\textbf{KKT (free-center in $\mathbb{R}^d$).}
For any optimal primal-dual pair $(Q,c,w)$ in $\mathbb{R}^d$,
\begin{equation}\label{eq:KKT-free}
Q^{-1} = M=\sum_i w_i(a_i-c)(a_i-c)^\top,\qquad
c=\frac{\sum_i w_i a_i}{\sum_i w_i},\qquad
\sum_i w_i=d,
\end{equation}
and complementary slackness: $w_i>0\Rightarrow (a_i-c)^\top Q(a_i-c)=1$.

\medskip\noindent
\textbf{KKT (origin-centered on the lift in $\mathbb{R}^{d+1}$).}
For any optimal primal-dual pair $(\widehat Q,\widetilde w)$ on the lifted points $\wta_i$,
\begin{equation}\label{eq:KKT-lift}
\widehat Q^{-1} = \widehat M =\sum_i \widetilde w_i \wta_i\wta_i^\top,\qquad
\sum_i \widetilde w_i=d+1,
\end{equation}
and complementary slackness: $\widetilde w_i>0\Rightarrow \wta_i^\top \widehat Q \wta_i=1$.

\medskip\noindent
\textbf{From $d$ to $d+1$.}
Let $w$ be optimal in $\mathbb{R}^d$ and set $c$ and $Q, M$ as in \eqref{eq:KKT-free}. Define the lifted matrix
$$
\widehat A:=\sum_i w_i \wta_i \wta_i^\top
=\sum_{i=1}^n w_i\,\wta_i \wta_i^\top = \begin{pmatrix} \sum_i w_i a_i a_i^\top & \sum_i w_i a_i\\ \sum_i w_i a_i^\top & \sum_i w_i \end{pmatrix}.
$$
The Schur complement using \eqref{eq:KKT-free} gives
\begin{equation}\label{eq:Mhat-inv}
\widehat A^{-1}=
\begin{pmatrix}
M^{-1} & - M^{-1} c\\
- c^\top M^{-1} & c^\top M^{-1} c + \frac{1}{d}
\end{pmatrix}
\quad\Longrightarrow\quad
\wta_i^\top \widehat A^{-1}\wta_i
=(a_i-c)^\top M^{-1}(a_i-c)+\frac{1}{d}.
\end{equation}
Then, we set 
$$
\widehat Q:=\frac{d}{d+1}\,\widehat A^{-1},\qquad
\widetilde w:=\frac{d+1}{d}\,w.
$$
Thus, we have $\sum_i \wtw_i=(d+1)/d\sum_i w_i=d+1$ and
$$
\widehat Q^{-1}=\frac{d+1}{d}\,\widehat A
=\sum_i \frac{d+1}{d}w_i\wta_i\wta_i^\top
=\sum_i \wtw_i \wta_i\wta_i^\top.
$$
By \eqref{eq:Mhat-inv} and feasibility of $(Q,c,w)$, we have for any point $a_i$ with $(a_i-c)^\top Q (a_i -c) \leq 1$,
$$
\wta_i^\top \widehat Q \,\wta_i
=\frac{d}{d+1}\left((a_i-c)^\top M^{-1} (a_i-c)+\frac{1}{d}\right)\leq 1,
$$
with equality whenever $w_i>0$ (since then $(a_i-c)^\top Q(a_i-c)=1$). Thus, primal feasibility and complementary slackness hold for $(\widehat Q,\wtw)$, which implies optimality.

\medskip\noindent
\textbf{From $d+1$ to $d$.}
Let $(\widehat Q,\wtw)$ be optimal in $\mathbb{R}^{d+1}$ and set
$$
w:=\frac{d}{d+1}\,\wtw,\qquad
c:=\frac{\sum_i \wtw_i a_i}{\sum_i \wtw_i}=\frac{\sum_i w_i a_i}{\sum_i w_i}.
$$
Define the matrix in $\R^{d\times d}$
$$
M:=\sum_i w_i(a_i-c)(a_i-c)^\top
\quad\text{and}\quad
Q:=M^{-1}.
$$
Let $\widehat A:=\sum_i w_i \wta_i\wta_i^\top$; then
$$
\widehat Q^{-1}=\sum_i \widetilde w_i \wta_i\wta_i^\top
=\frac{d+1}{d}\sum_i w_i\wta_i\wta_i^\top
=\frac{d+1}{d} \widehat A
\quad\Longrightarrow\quad
\widehat Q=\frac{d}{d+1}\,\widehat A^{-1}.
$$
Writing $\widehat A^{-1}$ in the block form \eqref{eq:Mhat-inv} yields
$$
\wta_i^\top \widehat Q \, \wta_i
=\frac{d}{d+1}\left((a_i-c)^\top M^{-1}(a_i-c)+\frac{1}{d} \right).
$$
Thus, since $\wta_i^\top \widehat Q \, \wta_i \leq 1$, we have $(a_i-c)^\top M^{-1}(a_i-c) \leq 1$.
So, $(Q,c)$ is primal feasible in $\mathbb{R}^d$. If $\widetilde w_i>0$, then
$\wta_i^\top \widehat Q \,\wta_i=1$ by the complementary slackness in $\R^{d+1}$, hence
$(a_i-c)^\top M^{-1}(a_i-c)=1$, which implies complementary slackness for $(Q,c,w)$ as well. Thus, $(Q,c,w)$ is an optimal solution in $\R^d$.
\end{proof}

\end{document}